\documentclass[Journal]{IEEEtran}
\ifCLASSINFOpdf
  \else
  \fi

\usepackage{amsmath,bm}
\usepackage{amsthm}
\usepackage{algorithm}
\usepackage{algorithmic}
\usepackage{graphics}
\usepackage{epsfig}
\usepackage{mathrsfs}
\usepackage{amssymb}
\usepackage[OT1]{fontenc}
\usepackage[usenames,dvipsnames,svgnames,table]{xcolor}
\usepackage{multicol}
\usepackage{epstopdf}
\usepackage[flushleft]{threeparttable}
\usepackage{caption}
\renewcommand{\AA}{\mathbb{A}}
\newcommand{\BB}{\mathbb{B}}
\newcommand{\MM}{\mathbb{M}}

\newcommand{\CC}{\mathbb{C}}
\newcommand{\ZZ}{\mathbb{Z}}
\renewcommand{\SS}{\mathbb{S}}

\newcommand{\LL}{\mathbb{L}}

\newcommand{\ba}{\mathbf{a}}
\newcommand{\bs}{\mathbf{s}}
\newcommand{\bw}{\mathbf{w}}
\newcommand{\bx}{\mathbf{x}}
\newcommand{\bA}{\mathbf{A}}

\newcommand{\floor}[1]{\left\lfloor #1 \right\rfloor}
\newcommand{\ceil}[1]{\left\lceil #1 \right\rceil}

\newtheorem{thm}{Theorem}

\newtheorem{lem}[thm]{Lemma}

\newtheorem{defn}[thm]{Definition}

\theoremstyle{remark}

\begin{document}
\title{On the Foundation of Sparse Sensing (Part II): Diophantine Sampling and Array Configuration}

\hyphenation{IEEE Transactions on Signal Processing}

\author{Hanshen Xiao, Beining Zhou, and Guoqiang Xiao

\thanks{Hanshen Xiao is with CSAIL and the EECS Department, MIT, Cambridge, USA. E-mail: hsxiao@mit.edu.}
\thanks{Beining Zhou is with Department of Computer Science, Stanford University, Stanford, USA. E-mail: cathyzbn@stanford.edu.}
\thanks{Guoqiang Xiao is with the College of Computer and Information Science, Southwest University, Chongqing, China. E-mail: gqxiao@swu.edu.cn.}
}

\maketitle

\begin{abstract}
In the second part of the series papers, we set out to study the algorithmic efficiency of sparse sensing. Stemmed from co-prime sensing, we propose a generalized framework, termed Diophantine sensing, which utilizes generic {\em Diophantine equation} theory and higher-order {\em sparse ruler} to strengthen the sampling time, the degree of freedom (DoF), and the sampling sparsity, simultaneously. 

With a careful design of sampling spacings either in the temporal or spatial domain, co-prime sensing can reconstruct the autocorrelation of a sequence with significantly denser lags based on Bézout theorem. However, Bézout theorem also puts two practical constraints in the co-prime sensing framework. For frequency estimation, co-prime sampling needs $\Theta((M_1+M_2)L)$ samples, which results in $\Theta(M_1M_2L)$ sampling time, where $M_1$ and $M_2$ are co-prime down-sampling rates and $L$ is the number of snapshots required. As for Direction-of-arrival (DoA) estimation, the sensors cannot be arbitrarily sparse in co-prime arrays, where the least spacing has to be less than a half of wavelength.

Resorting to higher-moment statistics, the proposed Diophantine framework presents two fundamental improvements. First, on frequency estimation, we prove that given arbitrarily large down-sampling rates, there exist sampling schemes where the number of samples needed is only proportional to the sum of DoF and the number of snapshots required, which implies a linear sampling time. Second, on Direction-of-arrival (DoA) estimation, we propose 
two generic array constructions such that given $N$ sensors, the minimal distance between sensors can be as large as a polynomial of $N$, $\Theta(N^q)$, which indicates that an arbitrarily sparse array (with arbitrarily small mutual coupling) exists given sufficiently many sensors. In addition, asymptotically, the proposed array configurations produce the best known DoF bound compared to existing coarray designs. 
\end{abstract}

\section{Introduction}
\noindent In the first part of the series papers, we have theoretically studied the necessary and sufficient (robust) condition for sparse sensing theory, characterized by (multiple) remaindering encoding problem. We connected and specified the relationship between Chinese Remainder Theorem method \cite{xiao2015error, xiao2016towards, tsp2018, sp2017notes, tvt2019, RingRCRT2020, sp2021, WCL2019} and co-prime sensing \cite{coprime1, coprime2}. On the other hand, we also pointed out that the deterministic necessary sampling constraint can be significantly relaxed at a cost of negligible failure, where co-prime sensing provides a concrete example.   

The study on co-prime sparse sensing, which was initialized in \cite{coprime1,coprime2}, spans almost last decade and has witnessed tremendous progress \cite{co-array1,pushing,qin,localization,co-array2,sp2020,sp2020_2,sp2019,sp2017}. It can be used in autocorrelation reconstruction based estimation with enhanced degrees of freedom (DoF), while the sparsity, in terms of sampling spacing either in time or space domain, is also preserved. The key idea of co-prime sensing is that, for two sequences $\mathcal{M}_1=\{m_1M_1\tau, m_1=0,1,2,...,2M_2-1\}$ and $\mathcal{M}_2=\{m_2M_2\tau, m_2=1,2,...,M_1-1\}$, where $M_1$ and $M_2$ are co-prime integers and $\tau$ is a temporal/spatial distance unit, the difference set of the pairs of elements from $\mathcal{M}_1$ and $\mathcal{M}_2$ will cover all consecutive multiples of $\tau$ starting from $-(M_1M_2-1)\tau$ to $(M_1M_2-1)\tau$. In the applications of frequency estimation, the two sequences, $\mathcal{M}_1$ and $\mathcal{M}_2$, represent the sampling instants in the time domain from two samplers; while for the case of Direction-of-arrival (DoA) estimation, $\mathcal{M}_1$ and $\mathcal{M}_2$ correspond to the sensor positions of two uniform linear arrays (ULA), respectively.

Such a carefully designed difference set from $\mathcal{M}_1$ and $\mathcal{M}_2$ can construct $O(M_1M_2)$ consecutive lags for autocorrelation estimation, which is commonly used in subspace estimation methods, such as MUSIC and ESPRIT \cite{music} etc. On the other hand, the two co-prime numbers $M_1$ and $M_2$ can be arbitrarily large, with a small failure specified in the first paper (Theorem 3). Enhanced sparsity and DoF suggest direct improvements over the physical restrictions of samplers, where fewer samplers of larger undersampling rates are allowed.  Such idea can be generalized to any integer ring \cite{li2019coprime}. Apart from the one-dimension case, high-dimension cases have also been explored \cite{twodimension-1,twodimension-2,multidimension}. 

Though co-prime sensing has been well studied, there are two practical problems still remaining open.
\begin{enumerate}
\item Can sparse sensing match the optimal sampling time with respect to resolution and snapshot needed. For frequency estimation, if $L$ snapshots are required for autocorrelation estimation, the time delay of co-prime sampling is $\Theta(LM_1M_2)$. When $M_1$ and $M_2$ are large enough, such delay could be unacceptable. Ideally, a linear sampling time is expected. 
\item Can an array configuration be arbitrarily sparse for DoA estimation? A closer placement of sensors will incur heavier mutual coupling. However, the minimal distance between sensors in existing co-prime arrays and (augmented) nested arrays  \cite{2014nested,super1,super2,super3,nest,2019nested}, is fixed to be less than a half of wavelength.
\end{enumerate}

In this paper, we answer both questions affirmatively. Before proceeding, we first formalize the theoretical foundation of the sparse sensing framework.

\noindent \textbf{Sparse Ruler:} Simply put, either in temporal or spatial domain, it is desired to take as few and sparse as possible sampling instants such that their difference set can produce as many as possible consecutive numbers. \footnote{Strictly speaking, distinct numbers, which are not necessarily being consecutive, are usually sufficient. Intuitively, the number of distinct numbers generated correspond to that of independent equations one can list regarding the parameters to be estimated. However, one may have to resort to more complicated compressed sensing methods for reconstruction \cite{unique2015, unique2019}, compared to the consecutive difference case. } Indeed, this is a classic combinatorics problem, which can be dated back to 1949 by Erdos and Gal \cite{erdos1948}, formally known as {\em sparse ruler}. Sparse ruler poses the following interesting question that, if we want a ruler to measure any distance between 1 to $K$, what is the least number of marks required? Clearly, it is not necessary to keep every mark varying from 0 to $K$, like a regular ruler. For example, when $K=10$, a better non-uniform mark selection exists, such as $\{0, 1, 2, 3, 6, 10\}$.  \footnote{With different motivations, there are also some other variants of sparse ruler, for example, Golomb ruler, dated back to 1930s \cite{golomb1953}, which requires that the differences between any marks have to be distinct.} Many asymptotic properties on the least number of marks in a sparse ruler have been studied in \cite{erdos1948, leech1956}, and so far, the best known concrete construction is shown in \cite{super3}, which is close to be optimal. With the above understanding, many practical concerns with respect to the sampling and estimation performance can be mathematically characterized within the sparse ruler model. Besides the number of consecutive differences produced, which usually determines the DoF and resolution, the other two challenges mentioned earlier regarding sampling time and sparsity can be equivalently described below.
\begin{enumerate}
    \item \textbf{Sampling Time}: In temporal sampling, it is expected that the points sampled (marks placed) are within a relatively small range, of which the maximal magnitude determines the sampling time required. 
    \item \textbf{Sensor Sparsity}: For spatial sampling (array configuration), it is expected that the minimal spacing between any two sensors (marks placed) to be as large as possible, which determines the coupling effect and the hardware requirement (chip size and circuit architecture) in practice. 
\end{enumerate}

\noindent \textbf{Contribution and Origination}: To tackle the above-mentioned problems, we switch to consider higher-order sparse ruler and apply generic Diophantine equation theory to provide concrete constructions. We provide two ways of such generalization. Our first result is on complex waveforms. In Section \ref{sec:complex_fre}, we take frequency estimation as an example and introduce the generic model of Diophantine sampling. For complex waveforms, we propose a novel idea to apply the third order statistics to estimate the second order auto-correlation. Such idea can be straightforwardly generalized to any higher orders. Applying three samplers, Theorem \ref{fre-sampling-thm} suggests that the proposed sampling scheme only requires $O(K+L)$ samples, where $K$ is the number of consecutive lags required and $L$ is the number of snapshots required. This matches the optimal sampling time asymptotically. We also describe a generic Diophantine sampling strategy with arbitrarily many samplers, concluded in Theorem \ref{fre_mul_est_thm}. Following that, in Section \ref{sec:complex_doa}, still restricted to complex waveform, we propose a third-order Diophantine array, which produces $\Theta(N^3)$ DoF and the minimal distance between sensors is $\Theta(N)$, specified in Theorem \ref{DoA}. 

From the beginning of Section \ref{sec:2q}, we consider another direction of generalization, which follows the $2q$-th order cumulant based method \cite{2q2011,2q2019}. We formally introduce the background and the concept of {\em $2q$-th  symmetric difference} in Section \ref{sec:2q}. In Section \ref{sec:shift-array}, we present a key architecture, termed {\em shifted array}, as the building block of higher-order co-array proposed.  Following that, in Section \ref{sec:4-order} and Section \ref{sec:6-order}, we provide two concrete construction frameworks for the forth-order and sixth-order scenario, respectively, which is also close to be optimal. Besides augmented DoF, such construction also strengthens the sparsity that the minimal spacing between any two sensors is $\Theta(N^2)$ and $\Theta(N^3)$, respectively. Taking those constructions as building blocks, we give a generic $2q$-th-order construction, which provides $O(17^{q/3}N^{2q})$ DoF. The simulation can be found in Section \ref{sec:sim}.  A summary of the results can be found in Table \ref{table_main}.

\begin{table*}
\caption{Summary of Diophantine Sampling and Array Proposed }
\begin{threeparttable}
\begin{tabular}{c | r| r| r| } 
\hline
\hline
\multicolumn{4}{c}{\textbf{{Diophantine Sampling}}}\\
\hline
\hline
Sampling Algorithm & Undersampling Rate Selection  & Sampling Time & Description and Analysis\\
 \hline
Three-sampler Sampling & $M_1 = 2+\Gamma, M_2 = 3+\Gamma, M_3 = 5 +\Gamma$  & $(2K+3L)(5+\Gamma)T_s$ & Theorem \ref{fre-sampling-thm} \\
\hline 
$N$-sampler Sampling & $M_i = i+\Gamma, i=1,2,...,N$ & $2(N-1)(K+ \frac{L\pi^2}{N(N-1)(N-2)})T_s$ & Theorem \ref{fre_mul_est_thm}\\
\hline
\hline
\multicolumn{4}{c}{\textbf{{Diophantine Array}}}\\
\hline
\hline
Array Configuration & Minimal Adjacent Sensor Spacing & DoF & Description and Analysis \\
 \hline
Third-order Array & $N\lambda/12$ & $N^3/27-1$ & Theorem \ref{DoA}   \\
Fourth-order Array &  $2(N/4-1/2)^2$ & $5(N/4-1/2)^4$ & Section \ref{sec:4-order}  \\
Sixth-order Array & $(N/6-1)^3$ &  $17(N/6-1)^6$ & Section \ref{sec:6-order}  \\
$2q$-th-order Array & $O((N/2q)^3)$ & $O(17^{\frac{q}{3}}(N/2q)^{2q})$  & Section \ref{sec:2q-general} \\
\hline 
\end{tabular}
  \begin{tablenotes}
      \footnotesize
      \item In Diophantine sampling, $\Gamma$ can be arbitrary positive integer. $K$ and $L$ represent the number of snapshot and lags required, respectively. $T_s = 1/f_s$, where $f_s$ is the Nyquist sampling rate.
      \item In Diophantine array, $N$ represents the number of sensors and $\lambda$ denotes the wavelength.
    \end{tablenotes}
\end{threeparttable}
  
\label{table_main}
\end{table*}

\section{Diophantine Equation Based Sampling}
\label{sec:complex_fre}
\subsection{Review of Co-prime Sampling}
\noindent We first flesh out the main idea of co-prime sampling. Let us consider a complex waveform formed by $D$ sources,
\begin{equation}
x(t) = \sum_{i=1}^{D} A_{i}e^{j(\omega_{i}t+\phi_{i})}+w(t) = \sum_{i=1}^{D} A_{i}e^{j\phi_{i}}e^{j\omega_{i}t}+w(t)
\end{equation}
where $\omega_{i} = 2\pi f_{i}T_{s}$ is the digital frequency and $T_s$ is the Nyquist sampling interval. $w(t)$ represents a Gaussian noise underlying of power $\sigma^2$.  $A_{i}$, $f_i$ and $\phi_{i}$ are the amplitude, the frequency and the phase of the $i$-th source, respectively. The phases $\phi_{i}$ are assumed to be random variables uniformly distributed in the interval $[0,~2\pi]$ and uncorrelated from each other. The two sampled sequences with sampling interval $M_1T_s$ and $M_2T_s$, respectively, are expressed as
\begin{equation}
\begin{aligned}
 &x_1[n] = x(nM_1T_s)+w_1(n)\\
 &x_2[n] = x(nM_2T_s)+w_2(n)
\end{aligned}
\end{equation}
where $w_1(n)$ and $w_2(n)$ are zero mean Gaussian white noise with power $\sigma^2$. 

With Bézout's identity, there exist $\{\alpha_1 = \langle M_1^{-1} \rangle_{M_2}, \beta_1 =  \langle -M_2^{-1} \rangle_{M_1} \}$, where $
\langle X \rangle_Y$ denotes the residue of $X$ modulo $Y$, such that $\alpha_1M_1- \beta_1M_2=1$. Thus, we can always find some solutions to the following equation by selecting ${\alpha}_k =  k{\alpha}_1$ and ${\beta}_k = k{\beta}_1$,
\begin{equation}
{\alpha}_kM_1- {\beta}_kM_2 = k
\end{equation}
Especially, in \cite{coprime1,coprime2}, it is shown that for $r \in \mathbb{Z}$, there exist $m^r_{1k} \in \{rM_2, rM_2+1, ... ,(r+2)M_2-1\}$ and ${m}^r_{2k} \in \{rM_1, rM_1+1, ... ,(r+1)M_1-1\}$ such that
\begin{equation}
\label{Bzout}
{m}^r_{1k}M_1- {m}^r_{2k}M_2 = k
\end{equation}
for $k \in \{-M_1M_2,...,0,1, ...,M_1M_2\}$. Furthermore, the autocorrelation of $x[n]$, which is a sampling sequence of $x(t)$ with a Nyquist sampling interval $T_s$, can be expressed as
\begin{equation}
\label{corre}
R_x[k] = \mathbb{E}_n\{x[n]x^*[n-k]\}=\sum_{i=1}^D A^2_i e^{j\omega_{i}k}. 
\end{equation}
According to equation (\ref{Bzout}), $R_x[k]$ can be also estimated by using the average of the products of those pairs $\{x_1[m^r_{1k}], x_2[m^r_{2k}]\}$\cite{autocorrelation2017,autocorrelation2015}. To this end, by rewriting $R_x[k]$ in the context of $x_1[m^r_{1k}]$ and $x_2[m^r_{2k}]$, i.e., $E_r\{x_1[m^r_{1k}]x^*_2[m^r_{2k}]\}$, we have Equation (\ref{haha}) below. (\ref{haha}) will be an unbiased autocorrelation estimation if the underline term is zero. In Part I (Theorem 3), we specify the condition. 
 
%that $\mathbb{E}_n\{x[n]x^*[n-k]\}=E_r\{x_1[m^r_{1k}]x^*_2[m^r_{2k}]\}$ in (\ref{haha}), 
\newcounter{mytempeqncnt}

\begin{figure*}[!h]
\normalsize
\hrulefill
\setcounter{mytempeqncnt}{\value{equation}}
\setcounter{equation}{5}
  \begin{equation}
  \label{haha}
\begin{aligned}
&\mathbb{E}_r[x_1[{m}^r_{1k}]x^*_2[{m}^r_{2k}]] =\mathbb{E}_r \{[ \sum_{i=1}^{D} A_ie^{j\phi_i}e^{j\omega_{i}{m}^r_{1k}M_1}+w_1({m}^r_{1k})][ \sum_{i=1}^{D} A_ie^{-j\phi_i}e^{-j\omega_{i}{m}^r_{2k}M_2}+w^*_2({m}^r_{2k}) ]\}\\
& = \mathbb{E}_r [ ~\sum_{i=1}^{D} A^2_i e^{j\omega_{i}({m}^r_{1k}M_1-{m}^r_{2k}M_2)}~]+\mathbb{E}_r[~\sum_{i \not = l}A_ie^{j\phi_i} e^{j\omega_{i}{m}^r_{1k}M_1} A_le^{-j\phi_l}e^{-j\omega_{l}{m}^r_{2k}M_2}~]+ \mathbb{E}_r [w_1({m}^r_{1k})w^*_2({m}^r_{2k})]\\
& = \mathbb{E}_r [ ~\sum_{i=1}^{D} A^2_i e^{j\omega_{i}k}] + \underline{\mathbb{E}_r[~\sum_{i \not = l} A_i A_l e^{j(\phi_i-\phi_l)} e^{j\omega_{l}k}  e^{j{m}^r_{1k}M_1( \omega_i -\omega_{l})}~]}  % = \sum_{i=1}^{D} A^2_i e^{j\omega_{i}k} = \mathbb{E}[x[n]x^*[n-k]]
\end{aligned}
\end{equation}
\setcounter{equation}{\value{mytempeqncnt}}
\hrulefill
\vspace*{4pt}
\end{figure*}
\setcounter{equation}{6}
If $L$ snapshots are used to estimate each $\mathbb{E}_n \{x[n]x^*[n-k]\}$, the delay of co-prime sampling is $\Theta(LM_1M_2T_s)$ \cite{complexity}. 
%\hrulefill
%\begin{figure*}[!h]
%\normalsize
%\setcounter{mytempeqncnt}{\value{equation}}
%\setcounter{equation}{5}
%    \begin{equation}
%          \label{corre}
%\begin{aligned}
%E[x[n]x^*[n-k]] &= E_n \{[ ~\sum_{i=1}^{D} A_ie^{j\CC_i}e^{j\omega_{i}n}+w(n)] [\sum_{i=1}^{D} A_ie^{-j\CC_i}e^{-j\omega_{i}(n-k)}+w^*(n-k)~ ]\} \\
%&= E_n [\sum_{i=1}^D A^2_ie^{j\omega_{i}k}] + E_n[\sum_{i \not = l} A_ie^{j\CC_i}e^{j\omega_{i}n}A_le^{-j\CC_l} e^{-j\omega_{l}(n-k)}] + E_n[w(n)w^*(n-k)]\\
%& = \sum_{i=1}^D A^2_ie^{j\omega_{i}k} + E_n [\sum_{i \not = l} A_iA_le^{j(\CC_i-\CC_l)}e^{j\omega_{l}k} e^{j(\omega_{i}- \omega_{l})n}]
%\end{aligned}
%\end{equation}
%\setcounter{equation}{\value{mytempeqncnt}}
%\hrulefill
%\vspace*{4pt}
%\end{figure*}
%\setcounter{equation}{6}
%It is noted that $w(n)$ is uncorrelated zero mean white noise sequences with power $\sigma_n^2$ and signals and noise are uncorrelated. Furthermore, we have
%\begin{equation}
%\begin{aligned}
%&E_n [\sum_{i \not = l} A_i A_le^{j(\CC_i-\CC_l)}e^{j\omega_{l}k} e^{j(\omega_{i}- \omega_{l})n}] \\
%&= \sum_{i \not = l} A_i A_le^{j(\CC_i-\CC_l)}e^{j\omega_{l}k}E_n [~ e^{j(\omega_{i}- \omega_{l})n}~] =0
%\end{aligned}
%\end{equation}
%Therefore, $E_n [x[n]x^*[n-k]] = \sum_{i=1}^D A^2_i e^{j\omega_{i}k}$.
\subsection{Diophantine Equation and Sampling}
\label{sec:two-sampler_dio}
\noindent In the Introduction section, we briefly mention the sparse ruler problem. Throughout the previous section, we have seen that the Bézout Theorem provides a concrete construction of the {\em mark selection}. The original sparse ruler problem and co-prime framework are both second-order based. A natural idea for generalization is to consider the higher-order difference. 

(Linear) Diophantine equation, which can be dated back to the 3rd century, yields the seed for the framework we consider here. In general, a linear Diophantine equation is in a form, 
\begin{equation}
\label{dio_generic}
    m_1M_1 + m_2M_2 + ... + m_tM_t = k, 
\end{equation}
for arbitrary $t$ integers $M_{[1:t]}$. The following lemma gives a necessary and sufficient condition that (\ref{dio_generic}) has a solution, which our following analysis will heavily rely on. 

\begin{lem}[\cite{hardy1979introduction}]
\label{hardy}
    The Diophantine equation (\ref{dio_generic}) has a solution if and only if $k$ is divided by $gcd(M_1,M_2, ..., M_t)$. 
\end{lem}

There are many ways to view (\ref{dio_generic}). In the context of temporal sampling, $M_{[1:t]}$ can represent the undersampling rates of $t$ samplers while $M_{[1:t]}$ can be the locations of sensors selected in an array. (\ref{dio_generic}) provides a formal way to consider the $t$-th order difference but also raises two natural questions,
\begin{enumerate}
    \item How to find solutions of (\ref{dio_generic}), where Lemma \ref{hardy} only captures the existence of solutions.
    \item How to use those generic difference for parameter reconstruction? So far, only the relationship between the second-order difference and auto-correlation is clear.  
\end{enumerate}
Our first result is to still rely on the auto-correlation reconstruction but via the higher-order difference (\ref{dio_generic}). We start by an auto-correlation reconstruction example with carefully designed third-order difference. 

Consider three samplers with undersampling rates $\{M_1, M_2, M_3\}$, respectively, and  we have 
\begin{equation}
\label{threedio}
m_1M_1 + m_2M_2 + m_3M_3 = k,
\end{equation}
Assume $gcd(M_1,M_2,M_3)=1$.
\footnote{Please note we do not require those three integers to be pairwise co-prime.}. We first present a generic semi-closed form of solutions $\{m_1, m_2, m_3\}$ in (\ref{threedio}). According to Lemma \ref{hardy}, there should exist two groups of integers, $\{a_1, a_2, a_3\}$ and $\{b_1,b_2,b_3\}$, such that
\begin{equation}
\label{Dio}
\left\{
            \begin{array}{lr}
a_1M_1+a_2M_2+a_3M_3 = 0 \\
b_1M_1+b_2M_2+b_3M_3 = 1 \\
 \end{array}
             \right.
\end{equation}
and the signs of $\{a_i\}$ are not the same, nor are the signs of $\{b_i\}$. This is because $M_{[1:3]}$ are all positive numbers larger than 1. Then, according to equation (\ref{Dio}) for $k=1,2,...,K$ and $l=1,2,...,L$, we have
\begin{equation}
\label{diosolution}
(kb_1+la_1)M_1 + (kb_2+la_2)M_2 +(kb_3+la_3)M_3 = k
\end{equation}
Without loss of generality, supposing $(kb_1+la_1)$ and $(kb_3+la_3)$ to be positive and $(kb_2+la_2)$ to be negative, we construct the following estimator
\begin{equation}
\label{auto-est-3-main}
    \mathbb{E}_{l}\big[{x}_1[kb_1+la_1]{x}^*_2[-(kb_2+la_2)]{x}_3[kb_3+la_3]\big]
\end{equation} 
to estimate the autocorrelation $R_x[k]$ at lag $k$, where ${x}_i[n] = x(M_inT_s) + w_i(n)$ denotes the sample sequence with a downsampling rate $M_i$. It is noted that, once $a_1M_1(\omega_i-\omega_v)+a_2M_2(\omega_u-\omega_v) \not = 0$ for $i\not =v \not = u \in \{1,2,...,D\}$, 
a negligible failure set, an unbiased estimation of the autocorrelation at lag $k$, $R_x[k]$, is produced as shown in (\ref{proof-main}). 
\begin{figure*}[!h]
\normalsize
\setcounter{mytempeqncnt}{\value{equation}}
\setcounter{equation}{10}
\begin{equation}
  \label{proof-main}
\begin{aligned}
&\mathbb{E}_l \{{x}_1[kb_1+la_1]{x}^*_2[-(kb_2+la_2)]{x}_3[kb_3+la_3]\} - R_x[k] \\
& = \sum_{i \not = u, v}  A_iA_uA_v e^{j[kb_1M_1\omega_i+kb_2M_2\omega_u+kb_3M_3\omega_v ]}e^{j(\phi_i-\phi_u+\phi_v)}\cdot \mathbb{E}_{l} \{e^{jl[a_1M_1(\omega_i-\omega_v)+a_2M_2(\omega_u-\omega_v)]}\} 
\end{aligned}
\end{equation}
\setcounter{equation}{\value{mytempeqncnt}}
\hrulefill
\vspace*{4pt}
\end{figure*}
\setcounter{equation}{11}

Thereby, (\ref{diosolution}) provides a {\em generic solution} to find proper sampling instants to estimate autocorrelation at arbitrary lag $k$. Consequently, the time delay of the proposed sampling scheme is determined by the maximal magnitude of the solution $m_i$ in (\ref{threedio}), i.e., $\max_{k,l,i} |kb_i+la_i|$, which can be $O((K+L)M)$ as proved in the following theorem with an arbitrarily large undersampling rate $M= \max \{M_1, M_2, M_3\}$. 
\begin{thm}
\label{fre-sampling-thm}
 When $a_1M_1(\omega_i-\omega_v)+a_2M_2(\omega_u-\omega_v) \not = 0$ for $i\not =v \not = u \in \{1,2,...,D\}$, there exist constants $c_1$ and $c_2$ such that the time delay of the proposed scheme is upper bounded by $(c_1K+c_2L)MT_s$, where $K$ is the number of consecutive lags and $L$ is the number of snapshots required to estimate  $\mathbb{E}_l \{{x}_1[kb_1+la_1]{x}^*_2[-(kb_2+la_2)]{x}_3[kb_3+la_3]\}$. Here, $M = \max\{M_1, M_2, M_3\}$.
\end{thm}
\begin{proof}
Consider a set of integers, say, $\{2,3,5\}$, which satisfies $gcd(2,3,5)=1$. Now, we try to find out two integer sets $\{a_1,a_2,a_3\}$ and  $\{b_1,b_2,b_3\}$ such that
\begin{equation}
\left\{
            \begin{array}{lr}
2a_1+3a_2+5a_3 = 0 \\
a_1+a_2+a_3=0  \\
2b_1+3b_2+5b_3 = 1 \\
b_1+b_2+b_3=0  \\
 \end{array}
             \right.
\end{equation}

We choose $\{a_1=2,b_1=1\}$, $\{a_2=-3,b_2=-2\}$, $\{a_3=1,b_3=1\}$ as solutions of (11). Due to $\sum_i a_i=0$ and $\sum_i b_i=0$, clearly they are also solutions to
\begin{equation}
(kb_1+la_1)(2+\Gamma) + (kb_2+la_2)(3+\Gamma)+(kb_3+la_3)(5+\Gamma) =k
\end{equation}
for any $\Gamma \in \mathbb{Z}$. Based on (12), we can get the lags to estimate the autocorrelation. Because $k \in \{1, 2, ..., K\}$ and $l \in \{ 1,2,...,L\}$, we have $\max_{i,k,l} |kb_i+la_i| \leq 2K+3L$.  Also,  $(kb_1+la_1)$ and $(kb_3+la_3)$, i.e., $(k+2l)$ and $(k+l)$, are always positive, while $(kb_2+la_2)$, i.e., $(-2k-3l)$, is negative. Thus, the total time delay is upper bounded by $\max_{i,k,l} |kb_i+la_i| \cdot \max_{i}M_iT_s \leq (2K+3L)(5+\Gamma)T_s$.
\end{proof}

The proof of Theorem \ref{fre-sampling-thm} is developed with a specific construction, where $M_{[1:3]}$ are selected as $\{M_1 = 2+\Gamma, M_2=3+\Gamma, M_3= 5+\Gamma\}$. Such a sampling strategy allows arbitrary large undersampling rates (with a sufficiently large $\Gamma$), while the number of sampling instants is independent of $\Gamma$ but only determined by $K$ and $L$, the number of lags and snapshots required, respectively, which suggests a linear sampling time $O((K+L)MT_s)$.

\subsection{Generalization to Multiple Samplers}
\noindent  In the following, we extend the above idea to a general framework provided multiple samplers. Given $N$ samplers, of which the sampling rates are denoted by $M_1, M_2, ... ,M_{N}$, respectively, a distributed Diophantine sampling can be naturally constructed by selecting any three of them and applying the scheme presented in Section \ref{sec:two-sampler_dio} as a building block.

With a similar reasoning, for a subgroup of triple samplers, say, $\{M_{i_1}, M_{i_2}, M_{i_3}\}$, $i_1, i_2, i_3 \in \{1,2,...,N\}$, we still first consider constructing two sets of solutions, $\{a_{i_1}, a_{i_2}, a_{i_3}\}$ and $\{b_{i_1}, b_{i_2}, b_{i_3}\}$, such that
\begin{equation}
\label{wanted}
\left\{
            \begin{array}{lr}
a_{i_1}M_{i_1}+a_{i_2}M_{i_2}+a_{i_3}M_{i_3} = 0 \\
a_{i_1}+ a_{i_2}+ a_{i_3}=0 \\
b_{i_1}M_{i_1}+b_{i_2}M_{i_2}+b_{i_3}M_{i_3}=1 \\
b_{i_1}+ b_{i_2}+ b_{i_3}=0 \\
\end{array}
             \right.
\end{equation}
which can be simplified to find out $a_{i_1},a_{i_3},b_{i_1},b_{i_3}$,
\begin{equation}
\label{feiyang}
\left\{
            \begin{array}{lr}
a_{i_1}(M_{i_1}-M_{i_2})+a_{i_3}(M_{i_3}-M_{i_2}) = 0 \\
b_{i_1}(M_{i_1}-M_{i_2})+b_{i_3}(M_{i_3}-M_{i_2}) = 1  \\
 \end{array}
             \right.
\end{equation}
In the following, we set to figure out how many triplet subgroups can be selected from $\{M_1, M_2, ... ,M_N\}$ such that the solutions of (\ref{feiyang}) exist. We select ${M_1, M_2, ... ,M_N}$ as the sequence of consecutive numbers starting from $1$ to $N$ shifted by some integer $\Gamma$, i.e, $M_i = i+\Gamma$. For any $M_i$, we consider the following sequence
\begin{equation}
\label{seq}
M_1-M_i, M_2-M_i, ... ,M_N-M_i
\end{equation}
and try to estimate the number of co-prime pairs among them. Since the sequence in (\ref{seq}) are still consecutive numbers ranging from $(-N,N)$, the number of primes among the sequence is upper bounded by $\pi(N)$, where $\pi(N)$ denotes the number of primes no bigger than $N$. Thus, by randomly picking any two of them, the probability that the two picked numbers are co-prime is lower bounded by
\begin{equation}
\prod_{j=1}^{\pi(N)} (1-\frac{1}{P^2_j}) > \prod_{j=1}^{\infty} (1-\frac{1}{P^2_j}) = \frac{6}{\pi^2}
\end{equation}
where $P_j$ is the $j$-th prime in the natural order and the above inequality follows from the density of primes \cite{hardy1979introduction}. Here, we use the fact that if we randomly select two numbers from $\mathbb{Z}$, the probability that they both share a prime factor $p_j$ is $\frac{1}{p^2_j}$. Therefore, we can totally find at least $\frac{6}{\pi^2}\binom{N}{3}$, i.e., $\frac{1}{\pi^2}N(N-1)(N-2)$, triplet sets such that $(M_{i_1}-M_{i_2})$ and $(M_{i_3}-M_{i_2})$ are co-prime. Consequently, there exist solutions satisfying (\ref{feiyang}), which is equivalent to that there are solutions to (\ref{wanted}).
Without loss of generality, we assume $M_{i_1}>M_{i_2}>M_{i_3}$, which implies $M_{i_1}-M_{i_2}>0$ and $M_{i_3}-M_{i_2}<0$ in (\ref{feiyang}). Hence, we can specifically set $a_{i_1}=M_{i_2}-M_{i_3}$, $a_{i_3}=M_{i_1}-M_{i_2}$, $b_{i_1}=\langle (M_{i_1}-M_{i_2})^{-1} \rangle_{(M_{i_2}-M_{i_3})}$ and $b_{i_2}=\langle (M_{i_3}-M_{i_2})^{-1} \rangle_{(M_{i_1}-M_{i_2})}$, which are all positive. Thus, both $a_{i_2}$ and $b_{i_2}$ should be negative due to the restrictions in (\ref{wanted}). Therefore, from (\ref{feiyang}), it is clear that both $|a_i|$ and $|b_i|$, if exist, are upper bounded by $2(N-1)$. Moreover, for $k=1,2,...,K$ and $l=1,2,...,L$, $(kb_{i_2}+la_{i_2})$ are always negative,whereas $(kb_{i_1}+la_{i_1})$ and $(kb_{i_3}+la_{i_3})$ are positive. We conclude the above by the following theorem.

\begin{thm}
\label{fre_mul_est_thm}
For arbitrary $N$ samplers, there exists a distributed Diophantine sampling scheme which can provide at least ${L}N(N-1)(N-2)/{\pi^2}$ virtual snapshots to estimate $\mathbb{E}_l \{{x}_{i_1}[kb_{i_1}+la_{i_1}]{x}^*_{i_2}[-(kb_{i_2}+la_{i_2})]{x}_{i_3}[kb_{i_3}+la_{i_3}]\}$ with time delay upper bounded by $2(N-1)(K+L)MT_s$. Here $M = \max_i\{M_i\}$
\end{thm}

\begin{proof}
As shown above, by selecting the undersampling rates in a form $M_i = \Gamma+i$, with an arbitrary $\Gamma$, one can find at least $N(N-1)(N-2)/{\pi^2}$ many triplet subsets such that $(\ref{feiyang})$ has solutions, where especially the largest value of the solution is upper bounded by $2(N-1)$. Therefore, applying Theorem \ref{fre-sampling-thm}, for any triplet set $\{M_{i_1}, M_{i_2}, M_{i_3}\}$ such that $(\ref{feiyang})$ has solutions, with a time delay at most $2(N-1)(K+L)M$, one can find $L$ snapshots for each of $K$ consecutive lags. With totally at least $N(N-1)(N-2)/{\pi^2}$ many such triplet sets, the claims follow obviously. 
\end{proof}

Before the end of this section, we have two remarks. First, it is worthwhile to mention that the proposed undersampling rate selection is close to be optimal. For any given $N$ undersampling rates, the number of triple subsets such that (\ref{wanted}) has solutions, is upper bounded by
\begin{equation}
\binom{N}{3} - \binom{N_e}{3} - \binom{N-N_e}{3} < \frac{N(N-1)(N-2)}{8}
\end{equation}
where $N_e$ is the number of even numbers among $N$ integers. This is because if $M_1, M_2$ and $M_3$ are all even or odd integers, (\ref{feiyang}) is not solvable. 

Second, throughout this section, we give a third-order based Diophantine sampling construction but it can be straightforwardly generalized to any higher order, once the solutions to (\ref{dio_generic}) can be efficiently found. We leave a systematical study to our future work.

\section{Diophantine Array for Complex Waveform}
\label{sec:complex_doa}
\noindent In this section, we proceed to generalize the idea in Diophantine sampling to array configuration for DoA estimation. As mentioned earlier, another important application of co-prime sensing is to provide enhanced freedoms for DoA estimation  \cite{radar2,radar3,radar1}. To be self-contained, we still first review the co-prime array and related sparse array construction. 

\subsection{Review of Co-prime and Nested Array} 
\noindent A co-prime array consists of two uniform linear arrays (ULA) with $M_1-1$ and $2M_2$ sensors respectively. The positions of the $(M_1-1)$ sensors in the first ULA are given by $\{M_2m_2d, m_2=1,2,...,M_1-1 \}$ and the positions of the other $2M_2$ sensors are given by $\{ M_1m_1d, m_1=0,1,...,2M_2-1\}$. Here, $d={\lambda}/{2}$ and $\lambda$ corresponds to the wavelength. As indicated by Bézout theorem, the difference set $\pm \{m_1M_1 - m_2M_2\}$ will cover all consecutive integers from $-(M_1M_2-1)$ to $M_1M_2-1$, which can further provide $(2M_1M_2-1)$ DoF. In general, there are two primary concerns in DoA estimation. The first is the number of consecutive lags to estimate autocorrelation. As shown in \cite{coprime1}, both the resolution and DoF are proportional to the number of the longest consecutive lags generated by the difference set \footnote{More discussions on the holes of the co-prime array can be found in \cite{unique2015, unique2019}. }. Second, a larger spacing amongst sensors is always desirable to reduce coupling effect.

Besides co-prime array, another important coarray construction is {\em nested array}, which is also formed by two ULAs. As a special case of co-prime array, the underlying idea is still based on Bézout theorem that any number $k$ varying from $[-M_1M_2, M_1M_2-1]$ can be represented by $k = \pm (m_2M_2-m_1)$, where $m_2 \in \{0, 1, ... , M_1\}$ and $m_1 \in \{0,1,...,M_2\}.$ In general, given the same number of sensors, nested array can produce about one time more DoF compared to the co-prime array. However, nested array also bears heavier coupling effect, since it has a much denser sub ULA with sensors consecutively located as $\{0, d, ... ,M_2d\}$ \cite{nest}. Variants of Nested array with improvements can be found in \cite{2014nested, 2019nested} and references therein. However, as mentioned earlier, the least pairwise sensor distance in existing sparse array is always set to be no bigger than a half of wavelength, $d = \lambda/2$.

\subsection{A Third-order based Diophantine Array }
\noindent In the following, we propose the first construction of Diophantine array. Analogous to Section \ref{sec:complex_fre}, we consider estimating the spatial auto-correlation with the third-order difference.

Consider a DoA estimation of $D$ signal sources with a linear array of $N$ sensors, where the position of each sensor is represented by 
$$ \SS = \{p_1 \cdot d, p_2 \cdot d, ... , p_N \cdot d \} .$$ Let $a_l(\theta_i) = e^{j(2\pi/\lambda) d_l \sin\theta_i}$ be the element of the steering vector corresponding to direction $\theta_i$. Assuming $f_c$ to be the center frequency of the band of interest, for narrow-band sources centered at $f_i+f_c$, $i=1,2,...,D$, the received signal being down-converted to baseband at the $l$-{th} sensor is expressed by
\begin{equation}
x_l(t) = \sum_{i=1}^D a_l(\theta_i)s_i(t)e^{2\pi f_i t}
\end{equation}
where $s_i(t)$ is a narrow-band source. When a slow-fading channel is considered, we assume $s_i(t)$ as some constant $s_i$ in a coherence time block \cite{FSF}. With the similar idea we use in frequency estimation, when $n_1+n_3 = n_2$,
\begin{equation}
\label{three-chanel}
\mathbb{E}_{n_1,n_3}\big[ x_{l_1}[n_1]x^*_{l_2}[n_2]x_{l_3}[n_3] \big] = \sum_{i=1}^{D} s^3_ie^{j(2\pi/\lambda) (p_{l_1}-p_{l_2}+p_{l_3} )\sin\theta_i},
\end{equation}
where $ x_{l_i}[n]= x_{l_i}(n)+w_{l_i}(n)$, $l_1, l_2, l_3 \in \{1,2,...,N\}$ with some independent zero-mean noise $w_{l_i}(n)$. Similarly, $$\mathbb{E}_{n_1,n_3}[x^*_{l_1}(n_1)x_{l_2}(n_2)x^*_{l_3}(n_3)] = \sum_{i=1}^{D} s^3_ie^{j(2\pi/\lambda) (-p_{l_1}+p_{l_2}-p_{l_3} )\sin\theta_i}.$$ 
With the above observation, we can now describe the Diophantine array configuration based on the following theorem. 

\begin{thm}
\label{DoA}
Assign $M_1=p_3p_1$, $M_2=p_3p_2$ and $M_3=p_1p_2$ such that $p_1$, $p_2$ and $p_3$ are relatively co-prime positive integers. Let sensors be located at $\{m_1M_1d, m_2M_2d, m_3M_3d ~|~ m_1 \in \{0,1,...,2p_2-1\}, m_2 \in \{0,1,...,p_1-1\}, m_3 \in \{0,1,...,p_3-1 \}\}$, $d = \lambda/2$, which form three sub ULAs, respectively. Then, the difference set $\{\pm (m_1M_1-m_2M_2) \pm m_3M_3\}$ contains all consecutive numbers starting from $-p_1p_2p_3+1$ to $p_1p_2p_3-1$.
\end{thm}

\begin{proof}
The proof is developed by applying Bézout theorem twice. We first consider the difference set between the first and the second sub ULA, where $m_1M_1-m_2M_2 = p_3(m_1p_1-m_2p_2)$, where $m_1 \in \{0,1,...,2p_2-1\}$ and $m_2 \in \{0,1,...,p_1-1\}$. Based on Bézout theorem, the difference set  $\pm\{m_1p_1-m_2p_2\}$ enumerates $ \{-p_1p_2+1,-p_1p_2+1, ... , p_1p_2-1\}$. Now, applying Bézout theorem again on the difference between the difference set of the first two sub ULAs and the third sub ULA, i.e., $\pm\{m_1M_1-m_2M_2\}$ and $\pm\{m_3M_3\}$, which essentially are the multiples of $p_3$ and $p_1p_2$, respectively, the triple difference set $\{\pm (m_1M_1-m_2M_2) \pm m_3M_3\}$ thus covers each integer at least starting from $-p_1p_2p_3+1$ to $p_1p_2p_3-1$.
\end{proof}

{Theorem \ref{DoA} shows that with $(p_1+2p_2+p_3-2)$ sensors, at least $(2p_1p_2p_3-1)$ DoF can be provided. Compared with conventional co-prime array based DoA estimation, given $(M_1+2M_2-1)$ sensors, the corresponding DoF is 
$$2M_1M_2-1 \leq (\frac{M_1+2M_2}{2})^2-1.$$ Moreover, the minimal distance between any two adjacent sensors is $\min \{q, p_1, p_2 \}$, since the positions of any two sensors share at least one common divisor from  $\{q, p_1, p_2 \}$. Thus, the minimal distance of sensor is $\Theta(Nd)$.}

On the other hand, to estimate the autocorrelation at lag $k$, we will find the snapshots at time $n_1, n_2$ and $n_3$, where $n_1+n_3=n_2$, from the three uniform subarrays, respectively. {Therefore, assuming that each sensor has collected $L$ samples, we can find ${L(L-1)}/{2}$ snapshots for autocorrelation estimation for each lag by the proposed strategy, instead of $L$ snapshots in co-prime arrays, though the computational complexity may slightly increase to construct those snapshots.} Thus, those additional samples can compensate for the heavier noise compromise in the third-order estimation (\ref{three-chanel}), which is less robust than the second-order based estimation in co-prime arrays.

\section{Background of $2q$-th-cumulant Based Array Processing}
\label{sec:2q}
\noindent From the start of this section, we present our second results of sparse array construction by connecting Diophantine equation and the $2q$-th-cumulant. 

\subsection{$2q$-th-cumulant Based DoA Estimation}
\noindent Consider a linear array with $N$ sensors, whose location set $\SS$ is denoted by 
	\begin{align}
	\SS&=\{p_1\cdot d,p_2\cdot d,\dots,p_N\cdot d\},
	\end{align}
with the spacing parameter $d={\lambda}/{2}$, and $\lambda$ is the wavelength of the signal. Suppose that there are $D$ independent narrowband signals $\{\bs_i(n)\mid i=1,2,\dots,D\}$ whose respective DoAs are $\theta_{[1:D]}=\{\theta_i\mid i=1,2,\dots,D\}$. Then, the array model for the received signal can be expressed by a vector,
	\begin{align*}
	\bx(n)&=\sum_{i=1}^D\ba(\theta_i)\bs_i(n)+\bw(n) =\bA(\theta)\bs(n)+\bw(n)
	\end{align*}
	where $\ba(\theta_i)$ is the steering vector with
	\[\ba(\theta_i)=[e^{-j\frac{2\pi p_1}{\lambda}\sin{\theta_i}},e^{-j\frac{2\pi p_2}{\lambda}\sin{\theta_i}},\dots,e^{-j\frac{2\pi p_N}{\lambda}\sin{\theta_i}}]^T\] $\bA(\theta)$ is a $N\times D$ steering matrix with $\bA(\theta)=[\ba(\theta_1),\ba(\theta_2),\dots,\ba(\theta_D)]$, the signal vector $\bs(n)=[\bs_1(n),\bs_2(n),\dots,\bs_D(n)]$, and $\bw(n)$ is independent white Gaussian. According to \cite{chev2q}, the $2q$-th order circular cumulant matrix, $C_{2q,x}(k)$, can be calculated as 
	\begin{align*}
		C_{2q,x}(k)&=\sum_{i=1}^Lc_{2q,\bs_i}[\ba(\theta_i)^{\otimes k}\otimes \ba(\theta_i)^{*\otimes q-k}]\\
		&\times[\ba(\theta_i)^{\otimes k}\otimes\ba(\theta_i)^{*\otimes(q-k)}]^H+\sigma_\bw^2I_{N^q}\delta(q-1)
	\end{align*}
	Here, $k$ serves as an index for the matrix arrangement with $0\leq k\leq q$ and $c_{2q, s_i}$ is the $2q$-th order circular autocumulant of a particular signal $\bm{s}_i$, i.e.,
	\[c_{2q, \bs_i}=\operatorname{Cum}\left[\bs_{i_1}(t), \cdots, \bs_{i_q}(t), \bs_{i_{q+1}}^{*}(t), \cdots, \bs_{i_{2q}}^{*}(t)\right]\] 
	Here,
	$\otimes$ represents the Kronecker product, $\{\cdot\}^*$ represents the conjugacy matrix, and $\{\cdot\}^H$ is the Hermitian transpose. $\{\cdot\}^{\otimes k}$ is defined as the $N^k\times1$ matrix where 
	\[\ba(\theta_i)^{\otimes k} =\underbrace{\ba\left(\theta_i\right) \otimes \ba\left(\theta_i\right)\otimes\cdots\otimes \ba\left(\theta_i\right)}_{k \text { times }}.\]

	This cumulant matrix is related to the $2q$-th order difference co-array. For the vector $\ba(\theta_i)^{\otimes k}$, its elements are expressed as 
	\[e^{j\frac{2\pi}{\lambda}(\sum_{i=1}^{l} p_{n_{i}}-\sum_{i=l+1}^{q} p_{n_{i}}) \sin\theta},\]
	where $\{n_1, n_2, ... , n_N\}$ is some permutation of $\{1, 2, ... , N\}$.  
	Each of these elements corresponds to a steering vector of a higher-order difference co-array that could estimate a signal's DoA. 
	
	\subsection{$2q$-th Symmetric  Difference set}
	\label{sec:higher-order-ruler}
	\noindent From the previous subsection, estimation with the $2q$-th-cumulant can also be described within the sparse ruler framework. Essentially, the difference set is defined by a weighted sum of $2q$ elements from $\SS$, of which a half weights are all set to 1 and the rest are -1, termed as {\em $2q$-th symmetric difference set} in the following. 
	
	\begin{defn}[$2q$-th symmetric difference set]
		The $2q$-th symmetric difference set $\CC^{2q}(\SS)$ of $\SS$ is defined as \[\CC^{2q}(\SS)=\{\sum_{i=1}^q p_{\pi(i)}-\sum_{i=q+1}^{2q} p_{\pi(i)} \mid \pi \in \Pi \}, \]
		where $\pi$ denotes a permutation across $[1:N]$ and $\Pi$ is the set of all $N!$ permutations.   
	\end{defn}
	
	In addition, to lighten the notations, for any two sets $\SS$ and $\SS'$, we use 
	$$\CC(\SS, \SS') = \{ p_i+p_j \mid p_i \in \SS, p_j \in \SS'\},$$
	to denote the {\em cross sum} of elements from $\SS$ and $\SS'$. In general, we can similarly define the {\em cross sum} of $m$ sets $\CC(\SS_1, \SS_2, ... ,\SS_m)$ by
	$$ \CC(\SS_1, \SS_2, ... ,\SS_m) = \{\sum_i p_i \mid p_i \in \SS_i\}. $$
	Especially, if one can write $\SS = \SS_1 \cup \SS_2 \cup ... \cup \SS_{2q}$, then 
	$$\CC(b_1\SS_1, b_2\SS_2, ... , b_{2q}\SS_{2q})  \subset  \CC^{2q}(\SS),$$
	for any selection of $b_{[1:2q]} \in \{1, -1\}^{2q}$ such that a half of $b_{[1:2q]}$ are $1$ and the rest half are $-1$. From a view of Diophantine equation (\ref{dio_generic}), the $2q$-th symmetric difference set puts a specific restriction on the solution whose Hamming weight is fixed to $2q$ and formed by $b_i$ with $q$ 1s and the other $q$ -1s. When we take each $\SS_i$ as a subarray of $\SS$, the $2q$-th symmetric difference set generated by $\SS$ can be made up of all possible combinations of those mutual sums with predetermined signs of its subarrays, and there are ${{2q}\choose{q}}$ of such combinations in total. 
	
	\subsection{Existing Works Reviews}
	\noindent The fourth-order cumulant based DoA estimation has been proposed and studied dated back to 1990s \cite{cardoso},\cite{chev4}, where higher-order moment statistics offer straightforward enhanced DoF and resolution. Such algorithms were later extended to the arbitrary {\em $2q$-th cumulant based MUSIC} \cite{chev2q}. By embedding of $q$ nested arrays, \cite{2q2011} gives the first concrete construction that provided $N$ sensors, one can produce $O(N^{2q})$ DoF compared to the best known $O(N^q)$ results in previous works. The underlying idea can be essentially summarized as the following. For $2q$ given numbers, $\{M_1, M_2, ... , M_{2q}\}$, one can construct a number representation system by taking $\{\prod_{i=1}^j M_i, j=0,1,...,2q-1\}$ as bases that 
	$$ k  = m_1 + m_2 M_1 + m_3 M_1M_2 + ... + m_{2q}\prod_{i=1}^{2q-1}M_i,$$
	for any $k \in [0,\prod_{i=1}^{2q}M_i)$ with $m_i \in [0,M_i)$. This gives a natural embedding of $2q$ subarray, each of which is formed by sensors at a spacing being multiple of $\prod_{i=1}^j M_i$. Most existing works follow this line:  \cite{exnested} and \cite{excoprime} proposed SAFOE-NA and SAFOE-CPA, which are extensions of nested arrays and co-prime arrays into the fourth-order cumulant structure. Recent improvements in their results include the EAS-NA-NA, EAS-NA-CPA in \cite{shift}, and an extension of the co-prime array in \cite{offgrid}. Various efforts have also been devoted to further generalize these structures to fit the arbitrary $2q$-th cumulant such as  \cite{SEML}, which improves the DoF in \cite{2q2011} by a factor $O(2^q)$. However, we have to point out that following this idea to simply embed nested or co-prime arrays in a recursion manner may encounter two main obstacles. First, with nested array as the building block, a dense subarrary is unavoidable. Second, it is noted that the cross sum amongst the $2q$ subarrays is exploited in a limited way. As mentioned before, ideally there will be ${{2q}\choose{q}}$ combinations, while only at most $2^q$ combinations are utilized in existing works \cite{2q2019}. To be formal, given $N$ sensors, the upper bound of such $2q$ symmetric difference is $\frac{N^{2q}}{(2q)!} \cdot \frac{(2q)!}{q!q!} = \frac{N^{2q}}{(q!)^2}.$ If we focus on the coefficient of the highest order, with $N^{2q}=(2q\bar{N})^{2q}$, we have 
\begin{equation}
    \label{upper_bound_2q}
\frac{(2q\bar{N})^{2q}}{(q!)^2} \approx \frac{(2q)^{2q}}{2\pi q(\frac{q}{e})^{2q}}\bar{N}^{2q} = \frac{(2e)^{2q}}{2\pi q}\bar{N}^{2q}.
\end{equation}

	\section{Shifted Array}
	\label{sec:shift-array}
    	\noindent In this section, we will present a key building block of our second results on Diophantine array based on $2q$-th-cumulant, termed as {\em shifted array}. One can apply this structure either in a nested or co-prime array by uniformly shift each subarray. A carefully designed shift, as shown later, can increase the spacing amongst sensors to reduce the mutual coupling, and on the other hand, it enables us to efficiently embed multiple subarrays to enhance the DoF.  In the following, we first take nested and co-prime array as two concrete examples and study their shifted versions in Section \ref{sec:shift-nest} and \ref{sec:shift-co}, respectively. Then, in Section \ref{sec:shift-co+nest}, we present a new framework by embedding shifted nested arrays into shifted co-prime arrays, which exploits the permutation invariant property of the cumulant method to achieve larger DoF.
	
	\subsection{Shifted Nested Array}
	\label{sec:shift-nest}
	
	\begin{defn}
		Let $\SS=\SS_1 \cup\SS_2$ be a linear array. Then, it is a shifted nested array if it has the form:
		\begin{align*}
		\SS_1&=\{(n_1N_2+\delta_1)d\mid n_1=0,1,2,\dots N_1\}\\
		\SS_2&=\{(n_2+\delta_2)d\mid n_2=0,1,2,\dots N_2\}
		\end{align*}
		We say that $\SS$ is shifted by factors of $\delta_1, \delta_2$.
	\end{defn}
	
	This shifted nested array has the following properties.
	
	\begin{lem}
	\label{lem:shift-nest}
		Let $\SS=\SS_1\cup\SS_2$ be a shifted nested array shifted by factors of $\delta_1,\delta_2$. Then
		\begin{align*}
		\CC(\SS_1,-\SS_2)&=\{\mu d\mid-N_2+\delta_1-\delta_2\leq\mu \leq N_1N_2+\delta_1-\delta_2\}\\
		&\supset\{\mu d \mid\delta_1-\delta_2\leq\mu\leq N_1N_2+\delta_1-\delta_2\}
		\\
		\CC(-\SS_1,\SS_2)&=\{\mu d\mid-N_1N_2-\delta_1+\delta_2\leq\mu\leq N_2-\delta_1+\delta_2\}\\
		&\supset\{\mu d \mid-N_1N_2-\delta_1+\delta_2\leq\mu\leq -\delta_1+\delta_2\}
		\\
		\CC(\SS_1,\SS_2)&=\{\mu d\mid\delta_1+\delta_2\leq\mu\leq N_1N_2+N_2+\delta_1+\delta_2\}\\
		&\supset\{\mu d\mid\delta_1+\delta_2\leq\mu\leq N_1N_2+\delta_1+\delta_2\}
		\end{align*}
	\end{lem}
	\begin{proof}
	 Let $\SS_1(n_1)$ and $\SS_2(n_2)$ denote the $n_1$-th and $n_2$-th sensor's position of the subarray $\SS_1$ and $\SS_2$ respectively. Then,
	\begin{align*}
	\SS_1(n_1)-\SS_2(n_2) & =(n_1N_2+\delta_1)d-(n_2+\delta_2)d\\
	&=(n_1N_2-n_2)d+(\delta_1-\delta_2)d
	\end{align*}
	\begin{align*}
	-\SS_1&(N_1-n_1)+\SS_2(N_2-n_2)\\
	&=-(N_1N_2-n_1N_2+\delta_1)d+(N_2-n_2+\delta_2)d\\
	&=(n_1N_2-n_2)d+(-N_1N_2+N_2-\delta_1+\delta_2)d\\
	\SS_1&(n_1)+\SS_2(N_2-n_2)\\
	&=(n_1N_2+\delta_1)d+(N_2-n_2+\delta_2)d\\
	&=(n_1N_2-n_2)d+(\delta_1+\delta_2+N_2)d
	\end{align*}
	Applying the same idea for the regular nested array \cite{nest}, the claim follows straightforwardly.
	\end{proof}
	Lemma \ref{lem:shift-nest} characterizes the consecutive co-difference in a shifted nested array, which is helpful when we handle the embedding in later sections.

	\subsection{Shifted Co-prime Arrays}
	\label{sec:shift-co}
	
	\begin{defn}
		Let $\SS=\SS_1\cup\SS_2$ be a linear array. Then, it is a shifted co-prime array if it has the form:
		\begin{align*}
		\SS_1&=\{(n_1+\delta_1)M_1d\mid n_1=0,1,2,\dots N_1\}\\
		\SS_2&=\{(n_2+\delta_2)M_2d\mid n_2=0,1,2,\dots N_2\}
		\end{align*}
		We say that $\SS$ is shifted by factors of $\delta_1, \delta_2$.
	\end{defn}
	\noindent The following lemma characterizes a property of the consecutive second-order difference produced in shifted co-prime arrays. 
	
	\begin{lem}
	\label{lem:shifted-co}
		Let $\SS_1\cup\SS_2$, $\SS_3\cup\SS_4$ be two co-prime arrays shifted by integer distances $\delta_1$, $\delta_2$ and $\delta_1-kM_2$, $\delta_2-kM_1$ for some integer $k$, respectively. $M_1,M_2$ satisfies that $\gcd(M_1,M_2)=1$. They can be expressed by the following,
		\begin{align*}
		\SS_1 &= \{(n_1+\delta_1)M_1d\mid n_1=0,1,2,\dots N_1\}\\
		\SS_2 &= \{(n_2+\delta_2)M_2d\mid n_2=0,1,2,\dots N_2\}\\
		\SS_3 &= \{(n_1+\delta_1-kM_2)M_1d\mid n_1=0,1,2,\dots N_1\}\\
		\SS_4 &= \{(n_2+\delta_2- kM_1)M_2d\mid n_2=0,1,2,\dots N_2\}
		\end{align*}
		Then,\[\CC(\SS_1\cup\SS_2)=\CC(\SS_3\cup\SS_4)\]
	\end{lem} 
	
	\begin{proof}
	 Because of the symmetry between $\CC(\SS_1,-\SS_2)$ and $\CC(\SS_2,-\SS_1)$ (and similarly for $\CC(\SS_3,-\SS_4)$ and $\CC(\SS_3,-\SS_4)$), we only consider one of them without loss of generality.
	\begin{align*}
	& \SS_3(n_1)-\SS_4(n_2)\\
	&=(n_1+\delta_1-kM_2)M_1d-(n_2+\delta_2-kM_1)M_2d\\
	&= (n_1+\delta_1)M_1d-(n_2+\delta_2)M_2d-kM_2M_1d+ kM_1M_2d\\
	&=(n_1+\delta_1)M_1d-(n_2+\delta_2)M_2d\\
	&=\SS_1(n_1)-\SS_2(n_2)
	\end{align*}
	\begin{align*}
	&\CC(\SS_1,-\SS_2)\\
	&=\{\SS_1(n_1)-\SS_2(n_2)\mid n_1=1,2,\dots N_1, n_2=1,2,\dots N_2)\}\\
	&=\{\SS_3(n_1)-\SS_4(n_2)\mid n_1=1,2,\dots N_1, n_2=1,2,\dots N_2\}\\
	&= \CC(\SS_3,-\SS_4)
	\end{align*}
	This proves the lemma.
	\end{proof}
	
	\subsection{Embedding of Shifted Nested and Co-prime Arrays}	
	\label{sec:shift-co+nest}
    \noindent In this section, we give a high-level picture of embedding shifted nested arrays into a shifted co-prime array to produce $2q$-th symmetric difference. First, we split an array $\SS$ into two sets denoted by $\mathbb{M}_1$ and $\mathbb{M}_2$. We use $\LL_{1}$ and $\LL_{2}$ to denote the {\em $q$-th order difference set} of $\mathbb{M}_1$ and $\mathbb{M}_2$, respectively. To be formal, $\LL_l$ is in a form, 
     \[  \{\sum_{i=1}^j p_{\pi(i)}-\sum_{i=j+1}^{q} p_{\pi(i)} \mid p_{\pi(i)} \in \mathbb{M}_l, j\in\{0,1,\dots q\} ,\pi \in \Pi \}, \] 
    for $l = 1,2$, which include all the possible combinations of $q$-element sum or difference from $\mathbb{M}_1$ and and $\mathbb{M}_2$, respectively.  As shown later, shifted nested array is a good candidate to construct those $q$-th order difference desired. 
	
	Second, we attempt to transform $\LL_1$ and $\LL_2$ to fit a co-prime array structure. We multiply each element in $\mathbb{M}_1$ by $M_1$ and those in $\mathbb{M}_2$ by $M_2$, with the greatest common divisor $\gcd(M_1,M_2)=1$. $\mathbb{M}_1$ and $\mathbb{M}_2$ could be separately formed by multiple nested subarrays and for each pair of subarrays $(\mathbb{M}_{1i},\mathbb{M}_{2j})$ with $\mathbb{M}_{1i}\in\mathbb{M}_1$ and $\mathbb{M}_{2j}\in\mathbb{M}_2$, they could be properly shifted according to Lemma \ref{lem:shift-nest}. 
	
	Finally, we set out to aggregate multiple co-prime arrays formed by each $\mathbb{M}_{1i}\cup \mathbb{M}_{2j}$. By Lemma \ref{lem:shifted-co}, we could then further shift such co-prime structures to maximize the consecutive $2q$-th symmetric difference set. To this end, we introduce the following lemma.
	
	\begin{lem}
	\label{lem:shift-product-coprime}
		Let $\SS_1$ and $\SS_2$ be subarrays of a co-prime array, where
		\begin{align*}
		\SS_1 &= \{n_1M_1d\mid n_1 = 0,1,2,\dots,N_1\}\\
		\SS_2 &= \{n_2M_2d\mid n_2 = \floor{\frac{-M_1}{2}}, \floor{\frac{-M_1}{2}}+1,\dots,0,\dots,\floor{\frac{M_1}{2}}\}
		\end{align*}
		with $N_1\geq M_2$. Then, the cross difference $\CC(\SS_1, \SS_2)$ generated is in a form,
		\[\{\mu d \mid -(N_1+1)M_1+\floor{\frac{M_1}{2}}M_2\leq\mu\leq(N_1+1)M_1-\floor{\frac{M_1}{2}}M_2\}\]
	\end{lem}
	
	\begin{proof}
	Without lost of generality, we only consider all the positive lags generated by $$\CC_+(\SS_1, -\SS_2) = \{\mid n_1M_1-n_2M_2\mid \}.$$ By the Bézout theorem, we can find a solution $(x_0, y_0)$ to the equation $\mid x_0M_1-y_0M_2\mid=\mu.$
	Then, this equation has a general solution \[\begin{cases}x=x_0+kM_2\\y=y_0+kM_1\end{cases}\text{for }k\in\ZZ\]
	Therefore, we can find a pair of solution $(x_1, y_1)$ such that $y_1\in[\floor{\frac{-M_1}{2}},\floor{\frac{M_1}{2}}]$. Here, we have $x_1 < N_1+1$ since or else, 
	\begin{align*}^{}
	\mid x_1M_1-y_1M_2\mid &= x_1M_1-y_1M_2 >(N_1+1)M_1-\floor{\frac{-M_1}{2}}M_2
	\end{align*}
	Let $n_1=x_1$ and $n_2=y_1$ if $x_1\geq0$, and $n_1=-x_1$ and $n_2=-y_1$ otherwise. Thus, we have found proper solutions of $\mid n_1p-n_2q\mid=\mu$ for any $\mu\in\CC_+(\SS_1, \SS_2)$.
	\end{proof}
	
	\section{Lower-order Shift Array Models}
		\subsection{The Fourth Order Array}
		\label{sec:4-order}
	\noindent We construct a fourth order difference co-array. Choose two co-prime integers $M_1, M_2$, $\gcd(M_1, M_2) = 1$, such that $M_1\leq N_1N_2$, $M_2 \leq N_3N_4$, with parameters $N_{[1:4]}$ which will be specified below. Consider an array in a form $\SS = \mathbb{M}_{11}\cup \mathbb{M}_{12}\cup \mathbb{M}_{21}\cup \mathbb{M}_{22}$, where
	\begin{align*}
	\mathbb{M}_{11}&=\{(n_1N_2+M_2)M_1d\mid n_1=0,1,2,\dots,N_1\}
	\\
	\mathbb{M}_{12}&=\{(n_2+\floor{\frac{M_2}{2}})M_1d\mid n_2=0,1,2,\dots, N_2\}
	\\
	\mathbb{M}_{21}&=\{(n_3N_4-\floor{\frac{M_1}{2}})M_2d\mid n_3=0,1,2,\dots,N_3\}
	\\
	\mathbb{M}_{22}&=\{(n_4-\floor{\frac{M_1}{2}})M_2d \mid n_4=0,1,2,\dots,N_4\}
	\end{align*}
	Then, the set of consecutive lags \[\CC^4(\SS)=\{\mu d \mid-M_{max}^4\leq\mu\leq M_{max}^4,\mu\in\ZZ\}\]
	where
	\begin{align*}
	M_{max}^4&=\begin{cases}\floor{\frac{5}{2}M_1M_2}, \text{ when $M_2$ is even}\\\\\floor{\frac{5}{2}M_1M_2}-M_2,  \text{ when $M_2$ is odd}\end{cases} \\
	& \leq\floor{\frac{5}{2}N_1N_2N_3N_4}
	\end{align*}
	
	We prove the above claim in the following. Due to the symmetry of the $2q$-th symmetric difference, we still only consider positive lags. First, we look at the combinations of the sum and the difference sets of $\mathbb{M}_{11}$ and $\mathbb{M}_{12}$ and $\mathbb{M}_{21}$ and $\mathbb{M}_{22}$ as sets of shifted nested arrays. Applying Lemma \ref{lem:shift-nest}, we have 
	$$
	\CC(\mathbb{M}_{11}, -\mathbb{M}_{12})=\{\mu M_1d \mid c_1 \leq\mu\leq N_1N_2+c_1 \}, $$
	$$\CC(\mathbb{M}_{11}, \mathbb{M}_{12})=\{\mu M_1d \mid c_2 \leq\mu\leq N_1N_2+c_2\}, $$
	$$ \CC(\mathbb{M}_{21}, -\mathbb{M}_{22})=\{\mu M_2d \mid c_3 \leq\mu \leq (N_3+1)N_4+c_3 \},$$
	$$ \CC(\mathbb{M}_{21}, \mathbb{M}_{22})=\{\mu M_2d \mid c_4 \leq\mu \leq (N_3+1)N_4+c_4 \},$$ where 
	\begin{equation}
	\begin{aligned}
	c_1
	&=\delta(\mathbb{M}_{11})-\delta(\mathbb{M}_{12})
	=\floor{\frac{-M_2}{2}}
	\\
	c_2
	&=\delta(\mathbb{M}_{11})+\delta(\mathbb{M}_{11})
	=\floor{\frac{3M_2}{2}}
	\\
	c_3
	&=\delta(\mathbb{M}_{21})-\delta(\mathbb{M}_{22})
	=0
	\\
	c_4
	&=\delta(\mathbb{M}_{21})+\delta(\mathbb{M}_{22})
	=-2\floor{\frac{M_2}{2}}
	\end{aligned}
	\label{shift-factor}
  	\end{equation}
    Next, still due to symmetry, we consider the following three virtual co-prime arrays which produce  fourth order symmetric difference in a form, $\CC(\MM_{11},-\MM_{12},\MM_{21},-\MM_{22})$, $\CC(\MM_{11},-\MM_{12},-\MM_{21},\MM_{22})$, and $\CC(\MM_{11},\MM_{12},-\MM_{21},-\MM_{22})$, respectively. To be clear, let $\SS_1 = \CC(\MM_{11},-\MM_{12})$, $\SS_2 = \CC(\MM_{11},\MM_{12})$, $\SS_3 = \CC(-\MM_{21},-\MM_{22})$, and $\SS_4 = \CC(\MM_{22},\MM_{22})$. Then, virtually, any $(\SS_i,\SS_j)$ for $i\in \{1,2\}$ and $j \in \{3,4\}$ forms a co-prime array and (\ref{shift-factor}) characterizes their respective shift factor. Since $M_1\leq N_1N_2$ and $M_2\leq N_4N_5$, by Lemma \ref{lem:shifted-co}, those three virtual co-prime arrays indeed together produce an extended coprime array in a form $\AA\cup \BB$, where 
	\begin{align*}
	\AA&=\{n_1M_1d\mid n_1=\floor{\frac{-M_2}{2}},\floor{\frac{-M_2}{2}}+1,\dots,\floor{\frac{M_2}{2}}\}\\
	\BB&=\{n_2M_2d\mid n_2=-2M_2 \text{ or } -2M_2+1,\dots,0\}
	\end{align*}
	
	By Lemma \ref{lem:shift-product-coprime},
	therefore, the number of total positive lags generated by $\AA\cup \BB$ is at least \[ \{\mu d\mid 0\leq\mu\leq \floor{\frac{5M_1M_2}{2}}\text{ or } \floor{\frac{5M_1M_2}{2}}-M_2\}.\]
	The equality is reached if and only if $M_1=N_1N_2$, $M_2=N_3N_4$ and when $q$ is even. The claim then follows clearly.
	
	The fourth order shifted array proposed successfully uses all ${4\choose2}$ permutations of the signs in the difference co-array. Overall, it yields a DoF bound of $O(\frac{5}{2}n^4)$ for its consecutive lag generation.
	
		\subsection{The Sixth Order Array}
		\label{sec:6-order}
        \noindent We construct a sixth order difference co-array as follows. Choose integers $M_1, M_2$ such that $\gcd(M_1, M_2) = 1$, $M_1\leq N_1N_2N_3$ and $M_2< N_4N_5N_6$. Consider the set of sensors positions, $\SS = \mathbb{M}_{11}\cup \mathbb{M}_{12}\cup \mathbb{M}_{13}\cup \mathbb{M}_{21}\cup \mathbb{M}_{22}\cup \mathbb{M}_{23}$ where
	\begin{align*}
	\mathbb{M}_{11}&=\{(n_1N_2N_3)M_1d\mid n_1=0,1,2,\dots,N_1\}
	\\
	\mathbb{M}_{12}&=\{(n_2N_3+M_2)M_1d\mid n_2=0,1,2,\dots, N_2\}
	\\
	\mathbb{M}_{13}&=\{(n_3+\floor{\frac{3M_2}{2}})M_1d\mid n_3=0,1,2,\dots,N_3\}
	\\
	\mathbb{M}_{21}&=\{(n_4N_5N_6-\floor{\frac{5M_1}{2}})M_2d\mid n_4=0,1,2,\dots,N_4\}
	\\
	\mathbb{M}_{22}&=\{(n_5N_6-\floor{\frac{7M_1}{2}})M_2d \mid n_5=0,1,2,\dots,N_5\}
	\\
	\mathbb{M}_{23}&=\{(n_6-5M_1)M_2d\mid n_6=0,1,2,\dots, N_6\}
	\end{align*}
	Then, the set of consecutive lags \[\CC^6(\SS)=\{\mu d \mid-M_{max}^6\leq\mu\leq M_{max}^6,\mu\in\ZZ\}\] where \[M_{max}^6=\floor{\frac{17}{2}M_1M_2}\leq\floor{\frac{17}{2}N_1N_2N_3N_4N_5N_6}\]
	\begin{proof}
	  Due to the symmetry, we still only consider the positive lags. First, we look at the following combinations of the sum and the difference sets of $\mathbb{M}_{12}$ and $\mathbb{M}_{13}$ and $\mathbb{M}_{22}$ and $\mathbb{M}_{23}$ as sets of shifted nested arrays. Applying Lemma 1, we get
	\begin{align*}
	\CC(\mathbb{M}_{12}, -\mathbb{M}_{13})&=\{\mu M_1\mid c^2_1 \leq\mu\leq N_2N_3+c^2_1,\mu\in\ZZ\}
	\\
	\CC(\mathbb{M}_{12}, \mathbb{M}_{13})&=\{\mu M_1\mid c^2_2 \leq\mu\leq N_2N_3+c^2_2,  \mu\in\ZZ\}
	\\
	\CC(\mathbb{M}_{22}, -\mathbb{M}_{23})&=\{\mu M_2\mid c^2_3 \leq\mu\leq N_5N_6+c^2_3,\mu\in\ZZ\}
	\\
	\CC(\mathbb{M}_{22}, \mathbb{M}_{23})&=\{\mu M_2\mid c^2_4 \leq\mu\leq N_5N_6+c^2_4,\mu\in\ZZ\}
	\end{align*} where
	$c^2_1 =\delta(\mathbb{M}_{12})-\delta(-\mathbb{M}_{13}) =\floor{\frac{-M_2}{2}}$, 
	$c^2_2 =\delta(\mathbb{M}_{12})+\delta(\mathbb{M}_{13}) =\floor{\frac{5q}{2}}$, 
	$c^2_3 = \delta(\mathbb{M}_{22})-\delta(-\mathbb{M}_{23}) =\floor{\frac{3M_1}{2}}$, and 
	$c^2_4 = \delta(\mathbb{M}_{22})+\delta(\mathbb{M}_{23}) =\floor{\frac{-17M_1}{2}}$.
	
	Next, we further embed one more additional subarray into the above cross sums. Applying Lemma \ref{lem:shift-nest}, we have 
	\begin{align*}
	\SS_1 &= \CC(-\mathbb{M}_{11},\mathbb{M}_{12},\mathbb{M}_{13})\\&=\{\mu M_1d \mid c^3_1 \leq\mu\leq N_1N_2N_3+c^3_1\}
	\\
	\SS_2 &= \CC(\mathbb{M}_{11},-\mathbb{M}_{12},\mathbb{M}_{13})\\&=\{\mu M_1d \mid c^3_2 \leq\mu\leq N_1N_2N_3+c^3_2\}
	\\
	\SS_3 &= \CC(\mathbb{M}_{11},\mathbb{M}_{12},-\mathbb{M}_{13})\\&=\{\mu M_1d \mid c^3_3\leq\mu\leq N_1N_2N_3+c^3_3\}
	\\
	\SS_4 &= \CC(-\mathbb{M}_{21},\mathbb{M}_{22},\mathbb{M}_{23})\\&=\{\mu M_2d \mid c^3_4 \leq\mu\leq N_4N_5N_6+c^3_4\}
	\\
	\SS_5 &= \CC(\mathbb{M}_{21},-\mathbb{M}_{22},\mathbb{M}_{23})\\&=\{\mu M_2d \mid c^3_5 \leq\mu\leq N_4N_5N_6+c^3_5 \}
	\\
	\SS_6 &= \CC(\mathbb{M}_{21},\mathbb{M}_{22},-\mathbb{M}_{23})\\&=\{\mu M_2d \mid c^3_6 \leq\mu\leq N_4N_5N_6+c^3_6\}
	\end{align*}
	where
	\begin{equation}
	\begin{aligned}
	c^3_1 & =-N_1N_2N_3-\delta(\mathbb{M}_{11})+c^2_2 =-N_1N_2N_3+\floor{\frac{5M_2}{2}}
	\\
	c^3_2 &=\delta(\mathbb{M}_{11})-c^2_1 = 0-\floor{\frac{-M_2}{2}} =\ceil{\frac{M_2}{2}}
	\\
	c^3_3 &=\delta(\mathbb{M}_{11})+c^2_1
	=0+\floor{\frac{-M_2}{2}} =\floor{\frac{-M_2}{2}}
	\\
	c^3_4 &=-N_4N_5N_6-\delta(\mathbb{M}_{21})+c^2_4 =-N_4N_5N_6-6M_1
	\\
	c^3_5 &=\delta({\mathbb{M}_{21}})-c^2_3 =-4M_1
	\\
	c^3_6 &=\delta(\mathbb{M}_{21})+c^2_3
	=\floor{\frac{-5M_1}{2}}+\floor{\frac{3M_1}{2}} =-M_1.
	\end{aligned} 
	\end{equation}
	Finally, we construct the sixth order symmetric difference by considering the cross sum amongst $\SS_{[1:6]}$. For notion simplicity, we denote $\CC_{CP}(\delta_a,\delta_b)$ as the lags of a coprime array generated by two subarrays, $\SS_a$ and $\SS_b$ where $\SS_a$ and $\SS_b$ are two coprime subarrays shifted by factors $\delta_a$, $\delta_b$ respectively. If we consider the cross sum between $\SS_1$ and $-\SS_4$, it is not hard to find that $\CC(\SS_1,-\SS_4)=\CC_{CP}(-N_1N_2N_3+\floor{\frac{5M_2}{2}},-N_4N_5N_6-6M_1).$
	Here, $\CC_{CP}(-N_1N_2N_3+\floor{\frac{5M_2}{2}},-N_4N_5N_6-6M_1)$ means virtually $\SS_1$ and $-\SS_4$, as two sub arrays, form a co-prime array shifted by $-N_1N_2N_3+\floor{\frac{5M_2}{2}}$ and $-N_4N_5N_6-6M_1$, respectively. Now, by Lemma \ref{lem:shifted-co}, we have the following observations,

	\begin{align*}
	\CC(\SS_1,-\SS_4)&=\CC_{CP}(-N_1N_2N_3+\floor{\frac{5M_2}{2}},-N_4N_5N_6-6M_1)\\
	&=\CC_{CP}(-N_1N_2N_3+\floor{\frac{M_2}{2}},-N_4N_5N_6-8M_1)
	\\
	\CC(\SS_1,-\SS_5)&=\CC_{CP}(-N_1N_2N_3+\floor{\frac{5q}{2}},-4M_1)\\
	&=\CC_{CP}(-N_1N_2N_3+\floor{\frac{M_2}{2}},-6M_1)
	\\
	\CC(\SS_1,-\SS_6)&=\CC_{CP}(-N_1N_2N_3+\floor{\frac{5M_2}{2}},-M_1)\\
	&=\CC_{CP}(-N_1N_2N_3+\floor{\frac{M_2}{2}},-3M_1)
	\end{align*}
	\begin{equation}
	\label{6-order-3-sub}
	\begin{aligned}
	\CC(\SS_2,-\SS_4)&=\CC_{CP}(\ceil{\frac{M_2}{2}},-N_4N_5N_6-6M_1)\\
	&=\CC_{CP}(\ceil{\frac{-M_2}{2}},-N_4N_5N_6-7M_1)\\
	\CC(\SS_2,-\SS_5)&=\CC_{CP}(\ceil{\frac{M_2}{2}},-4M_1)=\CC_{CP}(\ceil{\frac{-M_2}{2}},-5M_1)
	\\
	\CC(\SS_2,-\SS_6)&=\CC_{CP}(\ceil{\frac{M_2}{2}},-M_1)=\CC_{CP}(\ceil{\frac{-M_2}{2}},-2M_1)
	\\
	\CC(\SS_3,-\SS_4)&=\CC_{CP}(\floor{\frac{-M_2}{2}},-N_4N_5N_6-6M_1)
	\\
	\CC(\SS_3,-\SS_5)&=\CC_{CP}(\floor{\frac{-M_2}{2}},-4M_1)\\
	\CC(\SS_3,-\SS_6)&=\CC_{CP}(\floor{\frac{-M_2}{2}},-M_1)
	\end{aligned}
	\end{equation}
	Since $M_1 \leq N_1N_2N_3$ and $M_2 \leq N_4N_5N_6$, the nine sets produced in (\ref{6-order-3-sub}) together virtually form an extended co-prime array. The lags generate by all nine arrays is contained in the lags generated by $A\cup B$, where 
	\begin{align*}
	A&=\{n_1M_1d\mid n_1=\floor{\frac{-M_2}{2}},\floor{\frac{-M_2}{2}}+1,\dots,\floor{\frac{M_2}{2}}\}\\
	B&=\{n_2M_2d\mid n_2=-9M_1,-9M_1+1,\dots,0\}.
	\end{align*}
	
	By Lemma 3,
	Therefore, the number of total positive lags generated by $A\cup B$ is \[\CC(A\cup B) = \{\mu\mid 0\leq\mu\leq \floor{\frac{17M_1q}{2}}\}.\] Therefore, 
	\begin{align*}
	M_{max}^6&=\floor{\frac{17}{2}M_1M_2} \leq\floor{\frac{17}{2}N_1N_2N_3N_4N_5N_6}
	\end{align*}
	The equality is achieved if and only if $M_1=N_1N_2N_3$, $M_2=N_4N_5N_6$.
	\end{proof}
	In general, the sixth order NECP reaches an efficiency of ${17(N/6)^6}/{2}$ for lag generation. It utilizes $9$ out of $ {{6\choose2}}/{2}=15$ number of sign combinations.

	\section{$2q$-th Order Shifted Diophatine Array}
	\label{sec:2q-general}
	\noindent We first introduce the proposed layering technique. For any two arrays which can generate consecutive integer lags, we can manipulate them to form a higher-order difference co-array as follows. Suppose we have two arrays $\SS_1$ and $\SS_2$ in a form, $\SS_1=\{\alpha_1\cdot d,\alpha_2\cdot d,\dots,\alpha_{N_1}\cdot d\}$ and $\SS_2=\{\beta_1\cdot d,\beta_2\cdot d,\dots,\beta_{N_2}\cdot d\}$,   which produce $2q_1$-th order and $2q_2$-th order consecutive symmetric difference, respectively
	\begin{align}
	\label{2_q_1}
	\CC^{2q_1}(\SS_1)&=\{-\mu_1\leq\mu\leq\mu_1\}
	\end{align}
	\begin{align}
	\label{2_q_2}
	\CC^{2q_2}(\SS_2)&=\{-\mu_2\leq\mu\leq\mu_2\}
	\end{align}
	Then, we take the $2(q_1+q_2)$-th order difference co-array $\SS_1\cup\SS'_2$ where \[\SS'_2=\{2\beta_1\mu_1\cdot d,2\beta_2\mu_1\cdot d,\dots,2\beta_{N_2}\mu_1\cdot d\}\]
	This new co-array can generate consecutive lags \[\CC^{2(q_1+q_2)}(\SS_1\cup\SS'_2)=\{-2\mu_1\mu_2-\mu_1\leq\mu\leq2\mu_1\mu_2+\mu_1\}\]
	
	\begin{proof} Without lost of generality, we only consider positive lags. Consider any integer $\mu\in[-2\mu_1\mu_2-\mu_1,2\mu_1\mu_2+\mu_1]$. We have the fact that $\mu$ can be uniquely expressed as \[\mu=2k\cdot\mu_1+r\] for some integers $k\in[0,\mu_2]$ and $r\in[-\mu_1+1,\mu_1]$. Since $\CC^{2q_1}(\SS_1)$ and $\CC^{2q_2}(\SS_2)$ covers integer expressed in (\ref{2_q_1}) and (\ref{2_q_2}). Let $r\in\CC^{2q_1}(\SS_1)$ and $k\in\CC^{2q_2}(\SS_2)$. Therefore, there exist 
	\begin{align*}
	r&=\sum_{i=1}^{q_1} \alpha_{n_i}-\sum_{i=q_1+1}^{2q_1} \alpha_{n_i}, n_i\in[1,N_1]\\
	k&=\sum_{i=1}^{q_2} \beta_{n_i}-\sum_{i=q_2+1}^{2q_2} \beta_{n_i}, n_i\in[1,N_2]
	\end{align*}
	Then, we have found
	\begin{align*}
	\mu&=r+2k\cdot\mu_1\\
	&=(\sum_{i=1}^{q_1} \alpha_{n_i}-\sum_{i=q_1+1}^{2q_1} \alpha_{n_i})+2\cdot(\sum_{i=1}^{q_2} \beta_{n_i}-\sum_{i=q_2+1}^{2q_2} \beta_{n_i})\cdot\mu_1\\
	&=(\sum_{i=1}^{q_1} \alpha_{n_i}-\sum_{i=q_1+1}^{2q_1} \alpha_{n_i})+(\sum_{i=1}^{q_2} 2\beta_{n_i}\mu_1-\sum_{i=q_2+1}^{2q_2} 2\beta_{n_i}\mu_1)\\
	&\in\CC^{2(q_1+q_2)}(\SS_1\cup\SS'_2)
	\end{align*}
	\end{proof}
    
	Now, applying the layering technique shown above, we break up the $2q$-th order co-array structure into a couple of sixth-order shifted Diophantine array proposed in Section \ref{sec:6-order} . In the cases when $2q\equiv4$ or $2q\equiv2\pmod{6}$, we can then layer a fourth-order shifted array, shown in Section \ref{sec:4-order}, or a regular nested array, respectively. This method is efficient. In the case when $2q\equiv0\pmod{6}$, the such layering construction produce DoF in an order of
	\[O(2\cdot\frac{(2\cdot\frac{17}{2}(N/2q)^6)^\frac{q}{3}}{2}) = O(17^\frac{q}{3}(N/2q)^{2q}).\]
	Similarly, when $2q\equiv2$ or  $2q\equiv4\pmod{6}$, it produces $O(2\cdot17^\frac{q-1}{3}(N/2q)^{2q})$ and $O(5\cdot17^\frac{q-2}{3}(N/2q)^{2q})$ DoF respectively. As mentioned before, in comparison, the best known DoF bound of $2q$-th order array \cite{unique2019} in exisitng works is $O(2^q(N/2q)^{2q})$

\begin{figure}
   \centering 
 \caption{Third-order Diophantine Sampling and Array with Noise}
 \vspace{-0.1 in}
%	 \begin{minipage}[b]{0.5 \textwidth}
     \centering
   \includegraphics[width=0.5\textwidth]{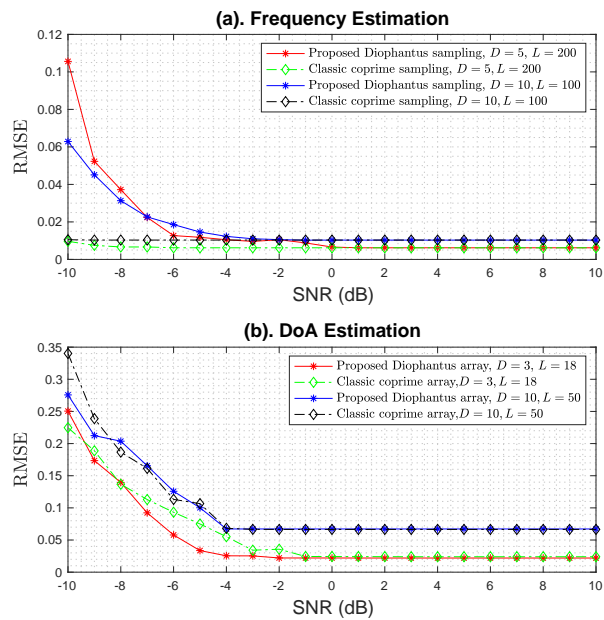}
   \vspace{-0.4 in}
 \label{sim}
\end{figure}
	
	\section{Simulation}
	\label{sec:sim}
	\subsection{Robustness}
    \noindent The first set of experiments we conduct is to measure the robustness of proposed schemes. We mainly focus on the novel third-order based Diophantine sampling and array proposed in Section \ref{sec:complex_fre} and   \ref{sec:complex_doa}, while the performance of $2q$-cumulant based method has been extensively tested in previous works such as \cite{2q2011, 2q2019}. We present the results of two numerical simulations in Fig. \ref{sim}, which compare their performances in the applications of frequency and DoA estimation with those of classic co-prime sensing.

    For frequency estimation, we randomly generate $5$ and $10$ frequencies respectively and set $K=L$. The proposed method is used to estimate the frequencies with three samplers of down-sampling rates $M_1=2+10^6$, $M_2=3+10^6$ and $M_3=5+10^6$. We average the root mean square error (RMSE) on 100 independent Monte-carlo runs with signal-to-noise ratio (SNR) ranged from -10 to 10dB. To evaluate the performance of the proposed method, we perform with Multiple Signal Classification (MUSIC) on attained virtual samples to estimate frequencies and the results are shown in Fig. 1(a). As expected, a small compromise in accuracy exists for proposed strategy since we use a higher third-order statistics. However, the time delay of co-prime sampling is around $10^{6}$ times longer than that of the Diophantine sampling.
    
    In the case of DoA estimation, we randomly generate $3$ and $10$ independent sources, respectively. For the third order Diophantine array proposed, we select $p_1=4$, $p_2=3$ and $p_3=5$, and thus totally $13$ sensors are used. We still apply MUSIC algorithm with $L=18$ and $L=50$ snapshots. As analyzed before, we can find $O(L^2)$ samples to estimate an autocorrelation for each lag $k$. The simulations are run 100 times and the averaged RMSEs are shown in Fig. 1(b). We can see that the proposed strategy in some cases is even with better performance than classic coprime array. Furthermore, our method provides up to $159$ DoF comparing to $57$ in the co-prime array. Also, the minimal spacing between sensors is $3d$, comparing to $d$ in all existing arrays.

\subsection{Array Configuration}
\noindent The second set of experiments we conduct is about the array configurations with a comprehensive comparison to existing sparse arrays. In the following, we use $\tau(j)$ to denote the number of adjacent sensor pairs whose corresponding spacing equals $jd$. First, we consider a third-order array shown in Section \ref{sec:complex_doa}, where we select $p_1=13$, $p_2=7$, $p_3=11$ with $N=36$ sensors in total. Such construction provides 2337 DoF, of which associated histogram of $\tau$ is shown in Fig. \ref{fig:3-order}. In comparison, we consider the regular nested array \cite{nest}, with $37$ sensors shown in Fig. \ref{fig:nested} and augmented nested array \cite{super3}, with $36$ sensors shown in Fig. \ref{fig:aug_nested}, which provide 687 and 865 DoF, respectively. Besides enhanced DoF, the proposed third-order Diophantine array is also of a greater sparsity with minimal spacing equaling $7d$ compared to $d$ in existing arrays.   

\begin{figure}[!htbp]
    \centering
    \includegraphics[width=0.5\textwidth]{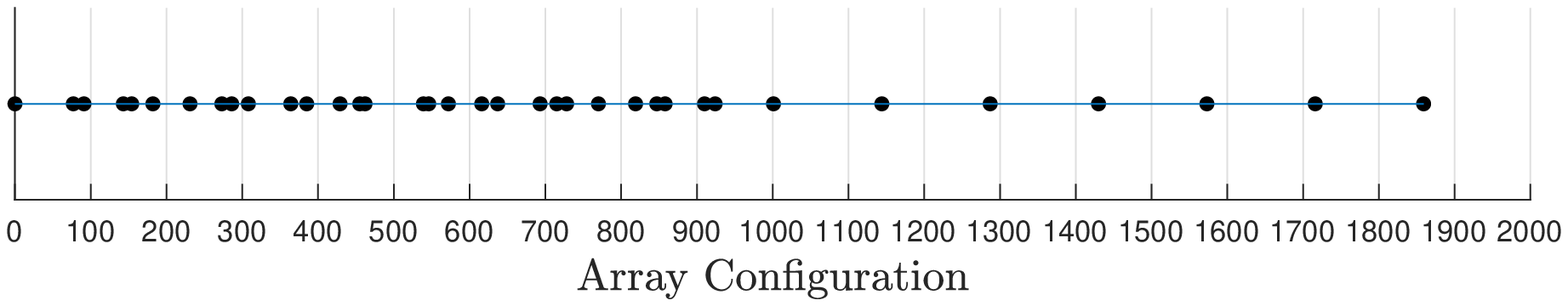}
    \includegraphics[width=0.5\textwidth]{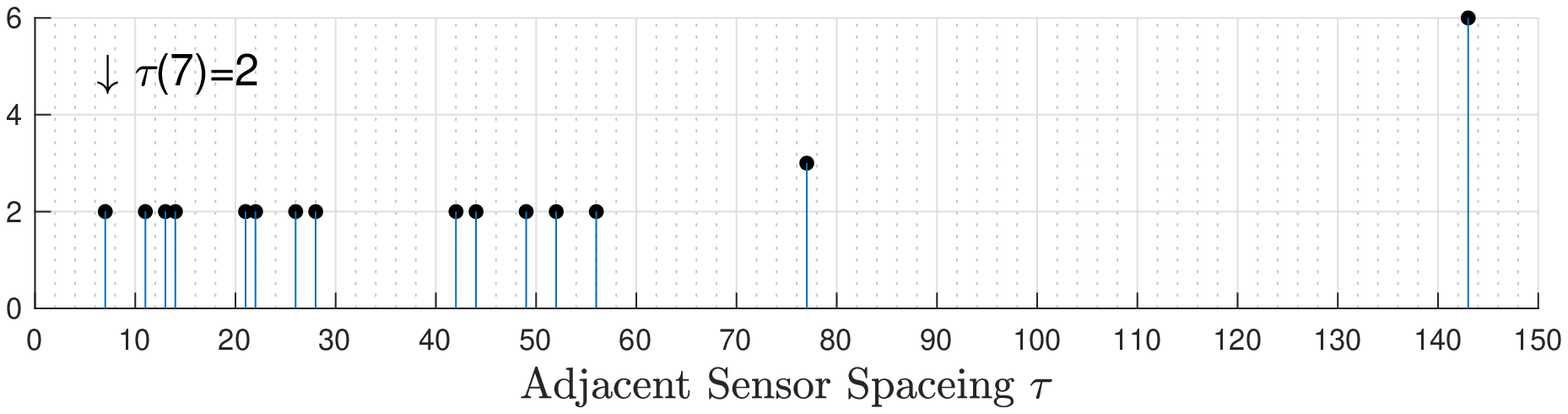}
    \caption{Third-order Diophantine Array}
    \label{fig:3-order}
        \vspace{-0.15 in}
\end{figure}

\begin{figure}[!htbp]
    \centering
    \includegraphics[width=0.5\textwidth]{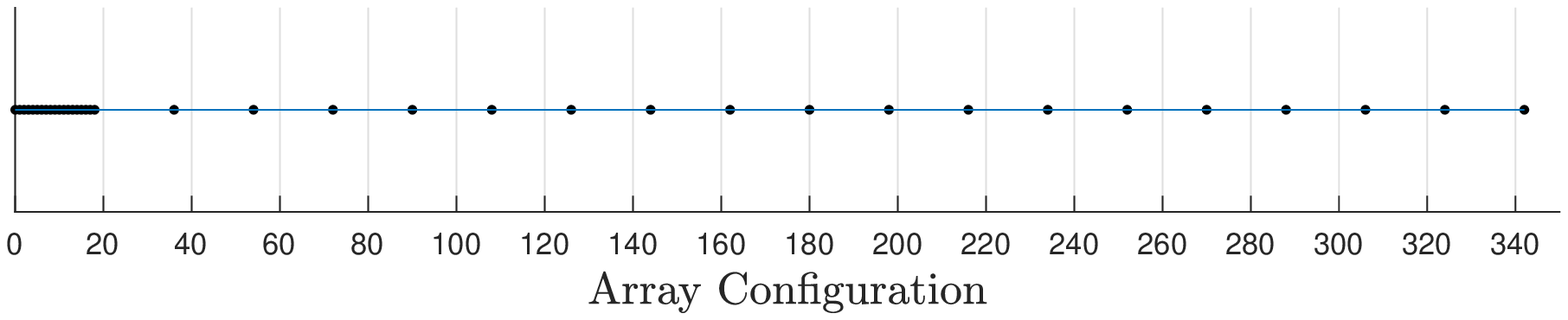}
    \includegraphics[width=0.5\textwidth]{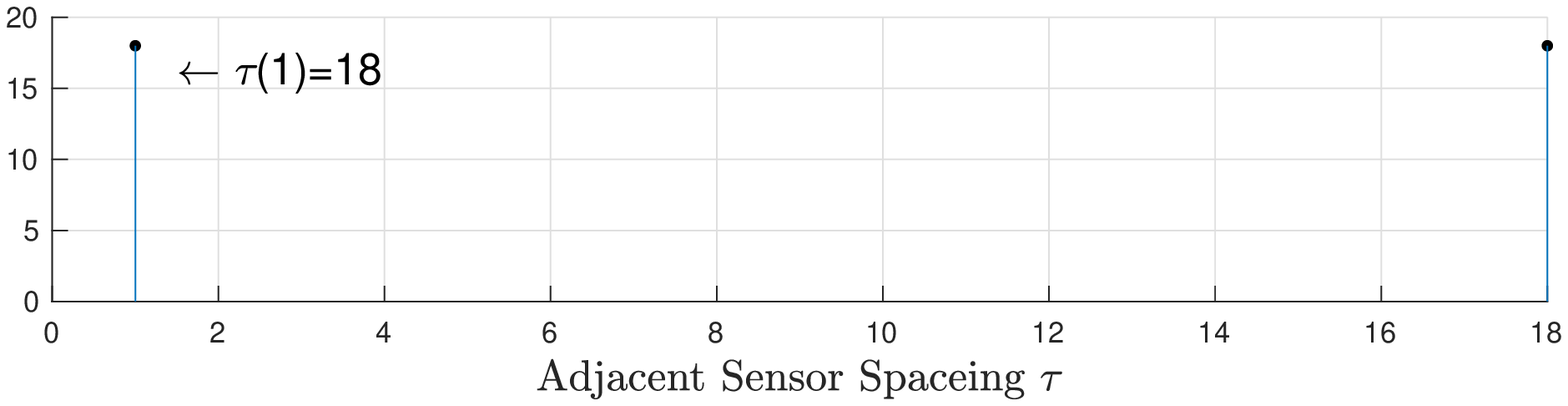}
    \caption{Regular Nested Array \cite{nest}}
    \label{fig:nested}
      \vspace{-0.15 in}
\end{figure}

\begin{figure}[!htbp]
    \centering
    \includegraphics[width=0.5\textwidth]{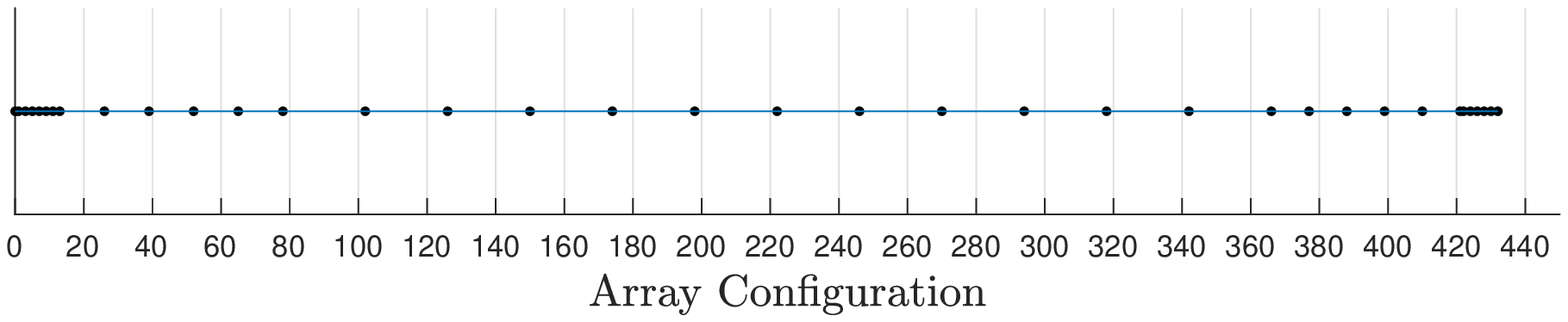}
    \includegraphics[width=0.5\textwidth]{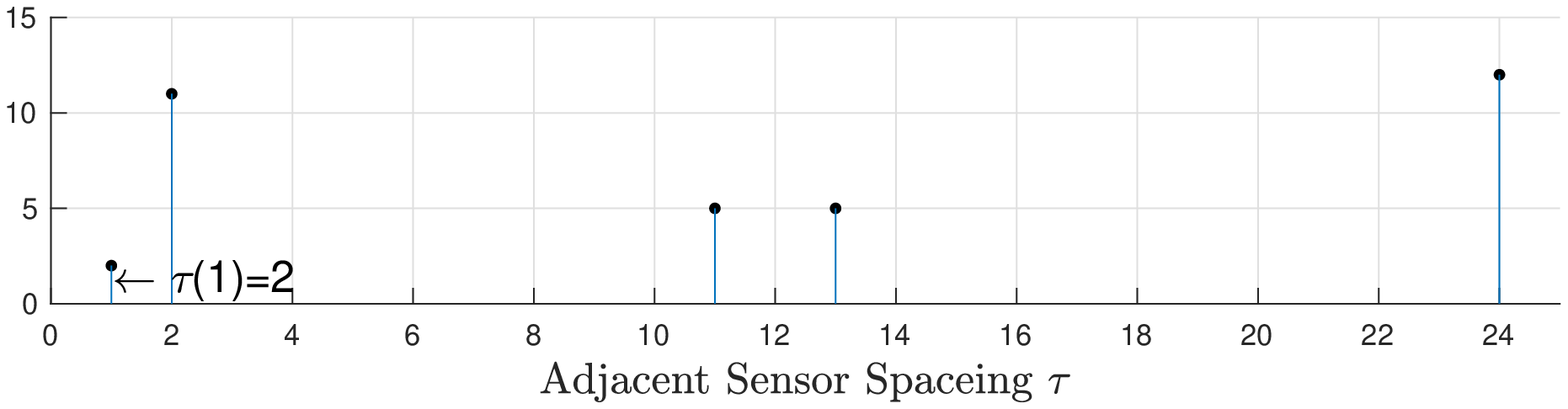}
    \caption{Augmented Nested Array II \cite{super3}}
    \label{fig:aug_nested}
        \vspace{-0.15 in}
\end{figure}

\begin{figure}[!htbp]
\label{fig:4-order}
    \centering
    \includegraphics[width=0.5\textwidth]{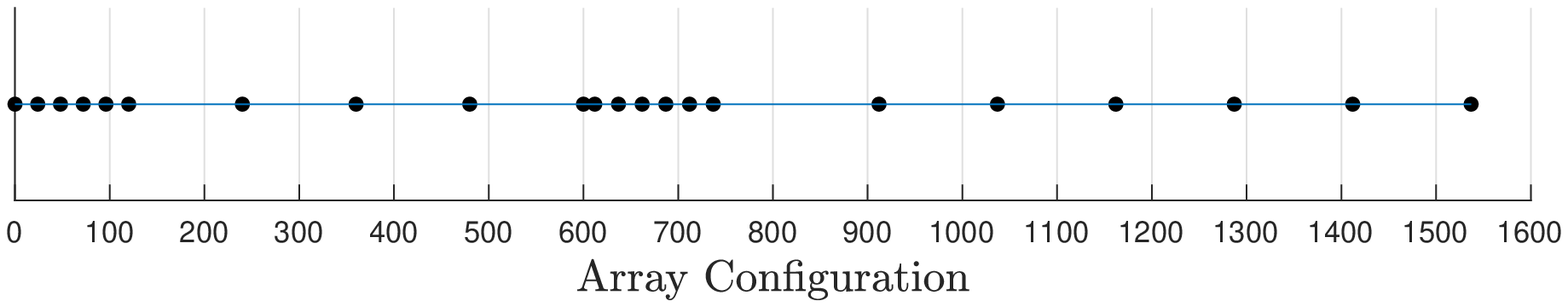}
    \includegraphics[width=0.5\textwidth]{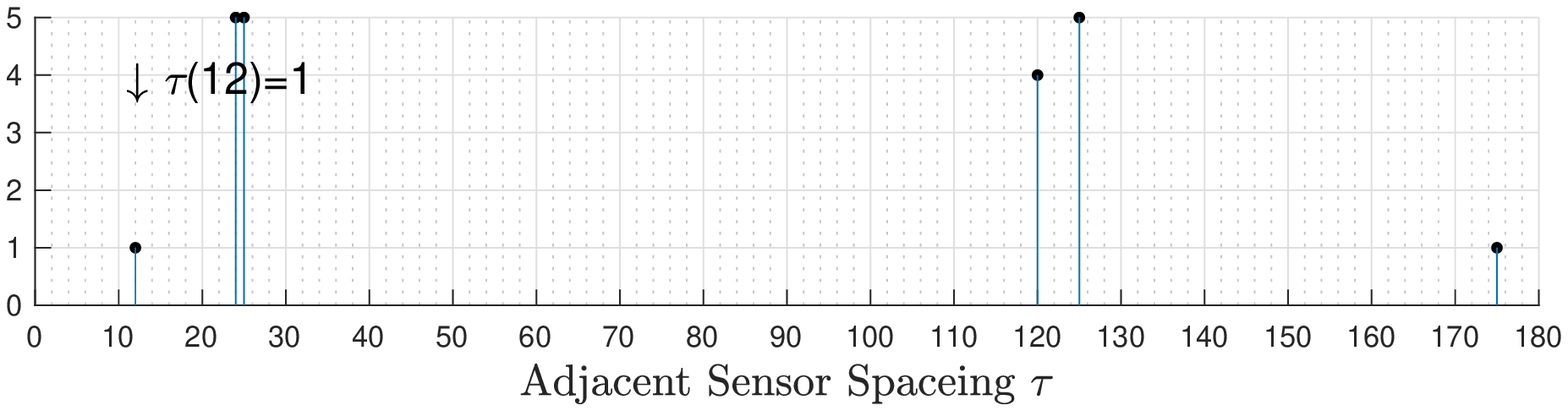}
    \caption{Fourth-order Shifted Array}
     \vspace{-0.1 in}
\end{figure}
\begin{figure}[!htbp]
\label{fig:6-order}
    \centering
    \includegraphics[width=0.5\textwidth]{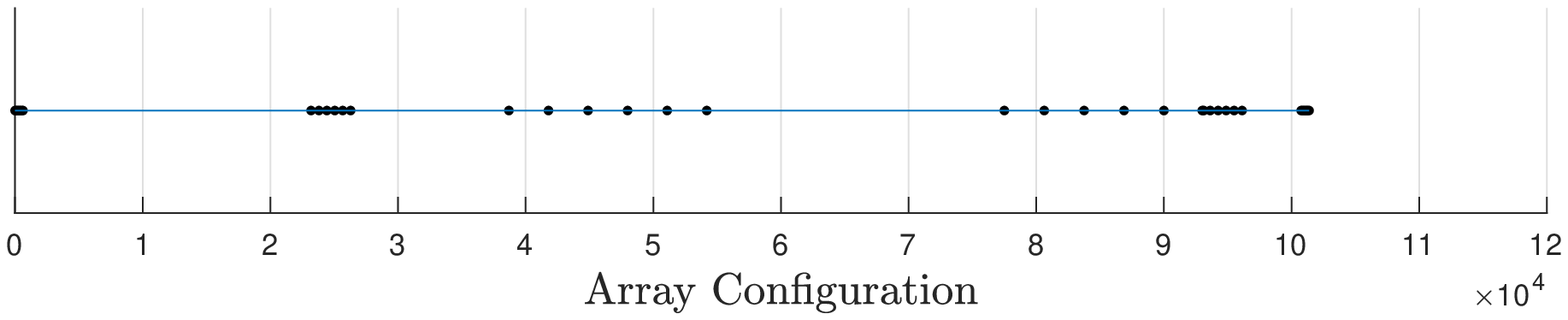}
    \includegraphics[width=0.5\textwidth]{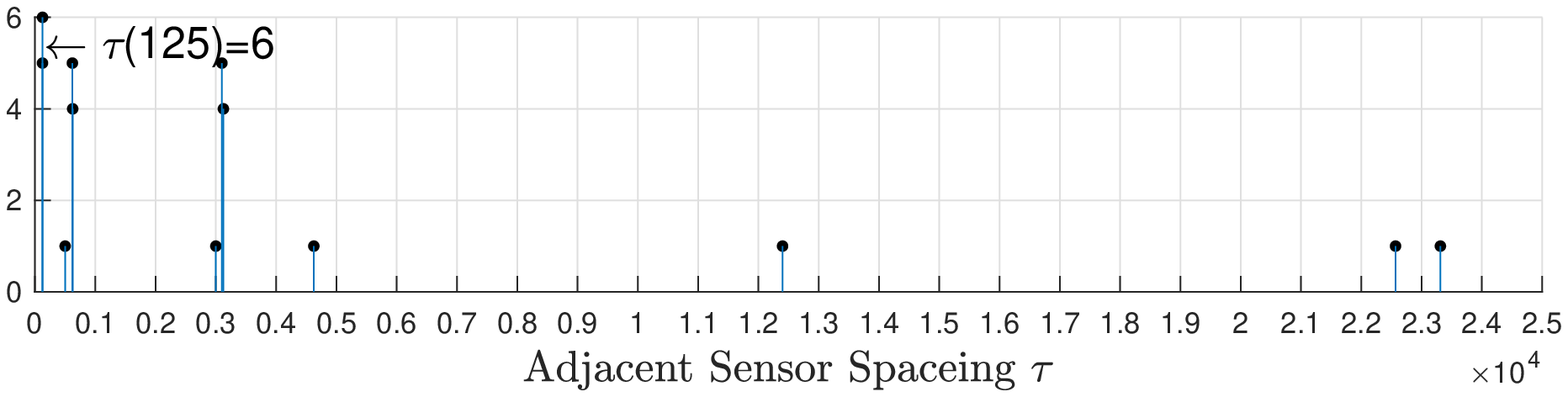}
    \caption{Sixth-order Shifted Array}
     \vspace{-0.15 in}
\end{figure}

\begin{figure}[!htbp]
    \centering
    \includegraphics[width=0.5\textwidth]{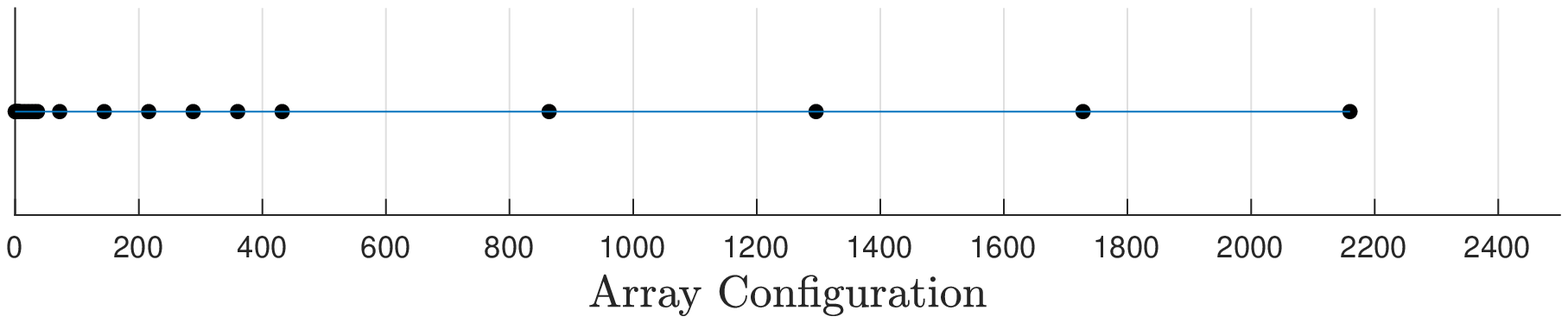}
    \includegraphics[width=0.5\textwidth]{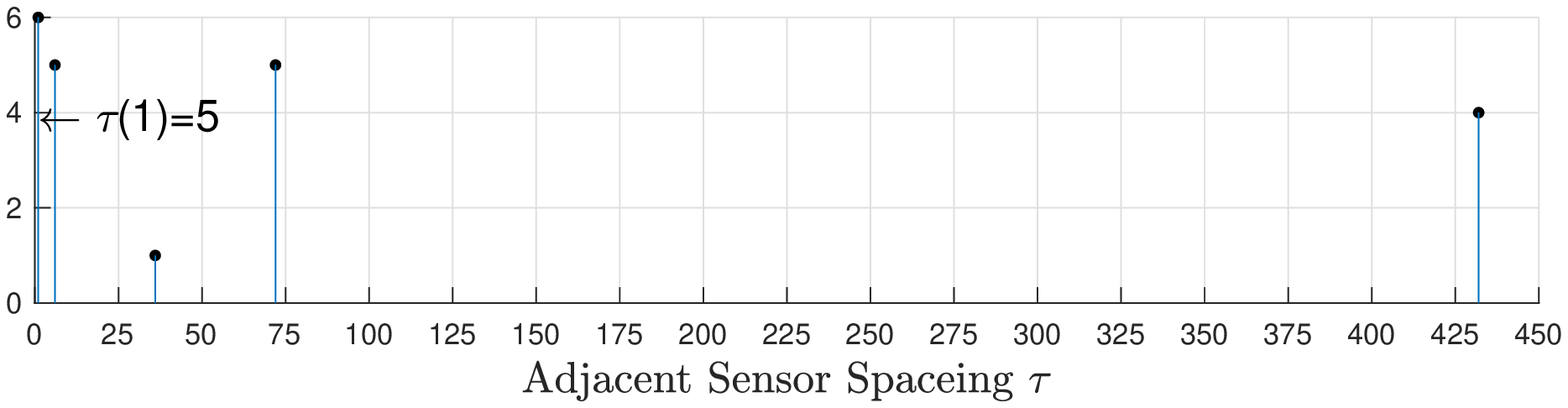}
    \caption{Fourth-order Improved Nested Array \cite{2q2019}}
    \label{fig:4-order-19}
     \vspace{-0.15 in}
\end{figure}

In the following, we consider the higher order cases. By uniformly selecting $N_i=5$, $i=1,2,3,4$, $M_1=25$ and $M_2=24$, we have a fourth-order shifted array defined in Section \ref{sec:4-order} with $N=22$ sensors. It provides 3445 DoF while the minimal sensor spacing equals $12d$. Similarly, by uniformly selecting $N_i=5$, for $i=1,2,...,6$, $M_1=125$ and $M_2=124$, we have a sixth-order shifted array defined in Section \ref{sec:4-order} with $N=36$ sensors and 271497 DoF. It is noted that the proposed fourth-order and sixth-order shifted Diophantine arrays are of $\Theta(N^2)$ and $\Theta(N^3)$ minimal spacing. As a comparison, we include the improved fourth-order nested array proposed in \cite{2019nested}, shown in Fig. \ref{fig:4-order-19}.

\section{Conclusion and Future Direction}
\noindent From a viewpoint of Diophantine equation and sparse ruler model, we present two sparse sensing frameworks to utilize higher-order statistics with stronger sampling efficiency. The techniques proposed in this paper shed light on how to efficiently coordinate multiple samplers either in temporal or spatial domain. Besides enhanced DoF, the Diophantine sampling and array proposed asymptotically match the optimal sampling delay and, especially for the array configuration design, it allows arbitrary sparsity given sufficiently many sensors. Meanwhile, we also leave several interesting open questions. For example on the minimal sensor spacing, whether for arbitrary $2q$-th order array, it can match $O(N^q)$. We only showed the possibility for the fourth and sixth order cases. 
		
\vspace{-0.1 in}

\bibliographystyle{ieeetr}
\bibliography{ref}

\end{document}